\documentclass[11pt,letter]{article}
\usepackage[utf8]{inputenc}
\usepackage{todonotes}
\usepackage{amsmath}
\usepackage{amsthm}
\usepackage{amssymb}
\usepackage{hyperref}
\usepackage{comment}
\usepackage[capitalise]{cleveref}
\usepackage[nocompress]{cite}
\usepackage{xspace}
\usepackage[shortlabels]{enumitem}
\usepackage{thmtools}
\usepackage{thm-restate}
\usepackage{algorithmic}
\usepackage[margin=1in]{geometry}

\theoremstyle{plain}

\newtheorem{lemma}{Lemma}  
\newtheorem{corollary}{Corollary} 
  
\newtheorem{fact}{Fact}

\newtheorem{observation}{Observation}

\theoremstyle{definition}
\newtheorem{definition}{Definition}

\newtheorem{example}{Example}

\usepackage[linesnumbered,lined,boxed,ruled,vlined]{algorithm2e}

\usepackage{thm-restate}

\newcommand{\Lap}{\mathrm{Lap}}
\SetKwInput{KwInput}{Input}
\SetKwInput{KwOutput}{Output}
\newcommand{\thresh}{\tau}

\newcommand{\alg}{\mathrm{Alg}}

\newcommand{\range}{\mathrm{range}}

\newcommand{\str}[1]{\mathrm{str}(#1)}

\newcommand{\setstr}{\mathcal{D}}

\newcommand{\counts}{\mathrm{count}}
\newcommand{\approxcount}{\counts^*}
\newcommand{\candidates}{\mathcal{C}}
\newcommand{\candidatespruned}{\mathcal{P}}
\newcommand{\diff}{\mathrm{diff}}
\newcommand{\candidatetrie}{T_{\candidates}}
\newcommand{\substring}[3]{#1\left[#2,~#3\right]}
\newcommand{\ST}{\mathrm{ST}}

\newcommand{\sd}{\mathrm{sd}}
\newcommand{\leaf}{\mathrm{leaf}}
\newcommand{\twokMin}{$2^k$-minimal\xspace}
\newcommand{\nil}{\textsf{NIL}}

\newcommand{\documentcount}{\textsc{Document Count}}
\newcommand{\substringcount}{\textsc{Substring Count}}
\newcommand{\substringmining}{\textsc{Substring Mining}}
\newcommand{\approxmining}{\textsc{$\alpha$-Approximate Substring Mining}}
\newcommand{\qgrammining}{\textsc{$q$-Gram Mining}}
\newcommand{\approxqgram}{\textsc{$\alpha$-Approximate $q$-Gram Mining}}

\title{Differentially Private Substring and Document Counting with Near-Optimal Error}

\date{}
\author{Giulia Bernardini \\ \texttt{giulia.bernardini@units.it} \and Philip Bille \\\texttt{phbi@dtu.dk} \and Inge Li G{\o}rtz \\\texttt{inge@dtu.dk} \and Teresa Anna Steiner \\\texttt{steiner@imada.sdu.dk}}

\begin{document}

\maketitle
\pagenumbering{gobble} 
\begin{abstract}
For databases consisting of many text documents, one of the most fundamental data analysis tasks is counting (i) how often a pattern appears as a substring in the database (\emph{substring counting}) and (ii) how many documents in the collection contain the pattern as a substring (\emph{document counting}). If such a database contains sensitive data, it is crucial to protect the privacy of individuals in the database. Differential privacy is the gold standard for privacy in data analysis. It gives rigorous privacy guarantees, but comes at the cost of yielding less accurate results. In this paper, we carry out a theoretical study of substring and document counting under differential privacy. 

We propose a data structure storing $\epsilon$-differentially private counts for all possible query patterns
with a maximum additive error of $O(\ell \cdot\mathrm{polylog}(n\ell|\Sigma|))$, where $\ell$ is the maximum length of a document in the database, $n$ is the number of documents, and $|\Sigma|$ is the size of the alphabet.  We also improve the error bound for document counting with $(\epsilon,\delta)$-differential privacy to $O(\sqrt{\ell} \cdot\mathrm{polylog}(n\ell|\Sigma|))$. We show that our additive errors for substring counting and document counting are optimal up to an $O(\mathrm{polylog}(n\ell))$ factor both for $\epsilon$-differential privacy and $(\epsilon,\delta)$-differential privacy.


Our data structures answer counting queries in time linear in the pattern length, occupy $O(n\ell^2)$ space, and can be constructed in $O(n^2\ell^4)$ time. When the query patterns are $q$-grams (i.e. have a fixed length $q$), the construction time can be drastically improved: in particular, a data structure storing $(\epsilon,\delta)$-differentially private $q$-gram counts can be constructed in $O(n\ell( \log q +\log |\Sigma|))$ time using $O(n\ell)$ space. Our data structures immediately lead to improved algorithms for related problems, such as privately mining frequent substrings and $q$-grams. Additionally, we develop a new technique of independent interest for differentially privately computing a general class of counting functions on trees.
\end{abstract}

\newpage
\pagenumbering{arabic}

\section{Introduction}
Differential privacy~\cite{Dwork2006} is the gold standard for privacy in data analysis. It has been extensively studied in theory and employed on a large scale by companies such as Google~\cite{googlekeyboard}, Apple~\cite{applescale}, and Uber~\cite{ubersql}, as well as public institutions such as the US Census Bureau~\cite{uscensus}. 
The goal of differentially private data analysis is to protect the privacy of any \emph{individual contribution}, while giving information about the \emph{general trends} in a database. This is often achieved as follows: Given a database from which we want to extract some information (specified through a \emph{query}), instead of giving the true output to the query, we add random noise to hide the contribution of any single data point. The random noise is scaled such that for any single data point, any given output should be \emph{roughly equally likely} whether the data point is present in the database or not. Specifically, we say that two databases are \emph{neighboring} if they differ in a single data point. An algorithm is $(\epsilon,\delta)$-differentially private if the probabilities of any output differ by at most a multiplicative factor of $e^{\epsilon}$ and an additive factor $\delta$ between any two neighboring databases. If $\delta=0$, we call the algorithm $\epsilon$-differentially private. The definition of differential privacy allows us to quantify \emph{how private} the output is. Promising differential privacy naturally comes with a loss of accuracy: a completely random output would be perfectly private, but not very useful.  The goal is to achieve a good tradeoff between privacy and accuracy.


One of the most fundamental types of problems in differential privacy is that of \emph{counting queries}, where a database consists of elements from a predefined universe, and the goal is, for a set of properties (i.e., a \emph{query family}), to count the number of elements in the database that satisfy each property. 
An example of a family of counting queries is that of histograms, where the query family contains a query for each element in the universe, and the goal is to count how often each element appears in the database.  Another example is the family of threshold functions, where each query corresponds to a threshold, and the goal is to count the number of elements below 
the threshold. The important question for any given query family is what is the best possible maximum additive error we can get when answering \emph{all} queries in the family while satisfying differential privacy. This is a challenging problem in general: In particular, for any $k$, there exists a family of $k$ counting queries such that an error of $\Omega(k)$ is necessary to satisfy $\epsilon$-differential privacy~\cite{DBLP:conf/stoc/HardtT10,DBLP:conf/pods/DinurN03}. However, the trade-off can differ significantly for various families of counting queries. See \cite{DBLP:books/sp/17/Vadhan17} for a comprehensive overview of the complexity of these problems. 

Many data analysis applications deal with databases of sequential data, which we refer to as \emph{documents} or \emph{strings}: For example, biological sequence data, web browsing statistics, trajectory data, and text protocols for next-word suggestions. Since these documents contain highly confidential information, an important question is whether such collections can be analyzed while preserving differential privacy. For databases consisting of many documents, we are interested in the \emph{patterns} they contain. In particular: For a database that is a collection of documents $\setstr=S_1,\dots, S_n$, and a query pattern $P$, how many documents in $\setstr$ contain $P$? This problem is called \documentcount. Another closely related question is: How often does $P$ appear as a substring of the documents in $\setstr$ \emph{in total}? We call this problem \substringcount. The difference between \documentcount\ and \substringcount\ is that in the latter, if a pattern occurs multiple times in the same document of $\setstr$, we count all its occurrences, while in the former, any document in $\setstr$ can contribute at most one to the count of any pattern. In the non-private setting, these problems have been extensively studied from a data-structure perspective, where the goal is to preprocess $\setstr$ to answer queries for any pattern $P$ efficiently: see e.g.~\cite{GKNP2013, FMMN2007, GK2011, Muthukrishnan2002, GNP2020, PJS1995, GM2005, LW2013, Navarro2014, MNSV2010, Sadakane2007, BHMM2009}.

In this paper, we carry out a theoretical study of the query families \documentcount\ and \substringcount\ under differential privacy. 
 While problems of this nature have been studied under differential privacy before from a practical perspective \cite{DBLP:conf/cikm/BonomiX13,DBLP:conf/sigmod/ZhangXX16,DBLP:conf/kdd/ChenFDS12,app12042131,DBLP:conf/securecomm/KhatriDH19,DBLP:conf/nips/KimGKY21,DBLP:conf/ccs/ChenAC12}, to the best of our knowledge, we are the first to provide a theoretical analysis on the additive error in terms of the problem parameters. Like in the non-private setting, we are interested in designing a data structure for these problems, which can then be queried for an arbitrary number of patterns while preserving differential privacy. The data structure should output a differentially private noisy count for any pattern query with the minimum possible additive error.

Trying to achieve this poses two main challenges. The first challenge is to provide a data structure that can be queried an arbitrary number of times without any privacy loss. Most differentially private algorithms are designed as one-shot algorithms, that is, they can only be run once on any given database: Indeed, running a differentially private algorithm several times
on the same database generally results in a higher privacy loss. As an example, consider the case where we want to output how many people in a database satisfy a given property. A differentially private algorithm computes the true output and adds random noise. Now, if we run the same differentially private algorithm on this database many times, then the average of all outputs will converge to the true output plus the mean of the random noise. Thus, at some point, we will be able to give the true answer with high probability, destroying any privacy guarantee. Summarizing, if one were to answer many queries on the same database by just running an independent differentially private algorithm each time, then the privacy loss grows with the number of queries.

The second challenge is the high dimensionality of the input data.
Recall that our privacy definition requires protecting the influence of any single document. This limits the accuracy we can hope for: As an example, consider the neighboring databases $\setstr$ and $\setstr'$, where the length-$\ell$ string $S_i=aaa\dots a$ of $\setstr$ is replaced by the length-$\ell$ string $S'_i=bbb\dots b$ in $\setstr'$.  The 
substring count of the query pattern $P=a$ differs by $\ell$ between $S_i$ and $S'_i$. Thus, intuitively, any algorithm computing the count of $P$ which does not significantly distinguish between $\setstr$ and $\setstr'$ must have an error of roughly $\ell$. Note that this lower bound holds for the problem of computing the count of a single (and very short) pattern! The high dimensionality also introduces another severe challenge: Note that $P=a$ may not even appear in $\setstr'$ at all - yet it must have a similar probability of producing a positive count as it does in $\setstr$. This means that assuming our input data consists of strings of a (known) maximum length $\ell$, then to satisfy $\epsilon$-differential privacy, \emph{any} string of length at most $\ell$ must have a positive count with some non-zero probability, not only those that actually appear as substrings in $\setstr$. On the other hand, we would like to efficiently construct a not-too-large data structure, and this prohibits explicitly computing a noisy count for all possible patterns.

We tackle the first challenge by giving a \emph{differentially private construction algorithm} for our data structure; since differential privacy is robust to post-processing, this data structure can then be queried ad-libitum without further privacy loss. Thus, preprocessing the database into a data structure has two advantages in this setting: Besides the speedup in queries, which is the usual motivation for data structures, it allows us to avoid 
the privacy cost of asking many queries.

To deal with the high dimensionality of string data, we use insights from string algorithms, along with a filtering of rare patterns at various steps of the construction algorithm. We combine this with a novel differentially private algorithm for counting on trees. The resulting algorithm produces a trie of $O(n\ell^2)$ nodes, where each node stores a noisy count for the string it corresponds to. For $\epsilon$-differential privacy, we show that the maximum error of these noisy counts is roughly $O(\ell)$ (up to polylogarithmic factors in the problem parameters) for both \documentcount\ and \substringcount. Strings that are not present in the trie have a count that lies below the error bound, with high probability in the coin flips of the algorithm. Surprisingly, this almost matches the discussed lower bound for computing the \substringcount\ of \emph{one} pattern (which we show holds for both $\epsilon$-differential privacy and $(\epsilon,\delta)$-differential privacy), even though the data structure can be used to compute the count of \emph{all possible} patterns.  Additionally, we show that when relaxing to $(\epsilon,\delta)$-differential privacy, we can further reduce the error for \documentcount\ by a factor $\sqrt{\ell}$, and this bound too is optimal up to polylogarithmic factors. 
Our strategy yields new results for differentially privately computing count functions on hierarchical tree data structures of certain properties, which may be of independent interest. As an example, our techniques give a differentially private data structure for the \emph{colored tree counting} problem.


As an immediate application of our work, we get improved results for problems which have been extensively studied under differential privacy from a practical perspective: namely, the problem of \emph{frequent sequential pattern mining}~\cite{DBLP:conf/cikm/BonomiX13,DBLP:conf/sigmod/ZhangXX16,DBLP:conf/kdd/ChenFDS12,app12042131,DBLP:conf/securecomm/KhatriDH19}, or the very similar problem of \emph{q-gram extraction}~\cite{DBLP:conf/nips/KimGKY21,DBLP:conf/ccs/ChenAC12}, where a $q$-gram is a pattern of a fixed length $q$. The goal in both of these problems is to analyze a collection of strings and report patterns that either occur in many strings in the collection (i.e., have a high document count) or appear often as a substring in the collection (i.e., have a high substring count). The private algorithms in these papers have been used for analysing transit data \cite{DBLP:conf/kdd/ChenFDS12} and publishing genome  data~\cite{DBLP:conf/securecomm/KhatriDH19}. However, existing papers do not give any theoretical guarantees on the accuracy of their algorithms in terms of the problem parameters. 
Using our data structures for \documentcount\ or \substringcount, we can solve mining problems by simply traversing the data structures and reporting all patterns with a noisy count above a given threshold. Notably, with our data structures it is possible to differentially privately mine frequent patterns from the same database using several frequency thresholds without any further privacy loss.

We show that the $\Omega(\ell)$ lower bound for $\epsilon$-differentially private substring and document count and for $(\epsilon,\delta)$-differentially private substring count also holds for the weaker problem formulation of outputting only the patterns whose count lies above a fixed threshold $\tau$ (where the error is defined as how strictly this threshold is kept, see Definition~\ref{def:approxmining}). Thus, the problem we study is a natural generalization of these mining problems. In the related work section, we argue that our tight bounds yield a polynomial improvement of the error compared to prior approaches.


\subsection{Setup and Results}

Our database is a collection $\setstr$ of documents (also called strings in this paper) of length at most $\ell$ drawn from an alphabet $\Sigma$ of size $|\Sigma|$. That is, we have $\setstr=S_1,\dots, S_{n}$, where $S_i\in\Sigma^{[1,\ell]}$ for all $i=1,\dots, n$. Thus, the \emph{data universe} is $X=\Sigma^{[1,\ell]}$ and any database is an element of $X^{*}$\footnote{Here ``$*$" is The Kleene operator: given a set $X$, $X^*$ is the (infinite) set of all finite sequences of elements in $X$.}. The order of the documents in $\setstr$ is not important, and we can also see $\setstr$ as a \emph{multiset}.
We call two databases $\setstr\in X^{*}$ and $\setstr'\in X^*$ \emph{neighboring} if there exists an $S'_i\in\Sigma^{[1,\ell]}$ such that $\setstr' = S_1,\ldots,S_{i-1}, S'_i, S_{i+1}, \ldots S_n$ for some $i$ (that is, document $S_i$ has been replaced with document $S'_i$).
We also denote the (symmetric) neighboring relation as $\setstr \sim \setstr'$. 
\begin{definition}[Differential Privacy]\label{def:privacy}
    Let $\chi$ be a data universe, $\epsilon>0$ and $\delta\geq 0$. An algorithm $A:\chi^{*}\rightarrow \range(A)$ is $(\epsilon,\delta)$-differentially private, if for all neighboring $\setstr, \setstr'\in \chi^*$ 
    and any set $U\subseteq \range(A)$, we have
    \begin{align*}
        \Pr[A(\setstr)\in U]\leq e^{\epsilon}\Pr[A(\setstr')\in U]+\delta.
    \end{align*}
\end{definition}
For $\delta=0$, the property from Definition~\ref{def:privacy} is also called \emph{$\epsilon$-differential privacy} or \emph{pure differential privacy}. For $\delta>0$, it is also called \emph{approximate differential privacy}.
\paragraph{Document and Substring Counting} 
Let $P$ and $S$ be strings over an alphabet $\Sigma$. We define $\counts(P, S)$ to be the number of occurrences of $P$ in $S$; if $P$ is the empty string, we define $\counts(P,S)=|S|$. For any $\Delta\in [1,\ell]$, let $\counts_{\Delta}(P,S)=\min(\Delta,\counts(P,S))$. We then define $\counts_{\Delta}(P,\setstr):=\sum_{S\in \setstr} \counts_{\Delta}(P,S)$ and $\counts(P,\setstr):=\counts_{\ell}(P,\setstr)$. Note that problem \documentcount\ uses $\Delta=1$; problem \substringcount\ uses $\Delta=\ell$.
In general, $\Delta$ limits the contribution of any one document (corresponding to one user) to the count of any pattern in the database.
Given a database $\setstr = S_1,\dots, S_{n}$, we want to build a data structure that supports the following types of queries: 
\begin{itemize}
        \item\documentcount($P$): Return the number $\counts_1(P,\setstr)$ of documents in $\setstr$ that contain $P$.  
        \item \substringcount($P$): Return the total number of occurrences $\counts(P,\setstr)$ of $P$ in the documents in $\setstr$.  
\end{itemize}

\begin{example}\label{ex:doc_substr_count}
Consider the database $\setstr=\{$\texttt{aaaa, abe, absab, babe, bee, bees}$\}$ and the string $P=\texttt{ab}$. Then $\counts_1(P,\setstr)=3$, $\counts(P,\setstr)=4$.
\end{example}

We want to construct a data structure to efficiently answer \documentcount\ and \substringcount\ queries while preserving differential privacy. For this, we need to introduce an error. We say that a data structure for $\counts_{\Delta}$ has an \emph{additive error} $\alpha$, if for any query pattern $P\in \Sigma^{[1,\ell]}$, it produces a value $\approxcount_{\Delta}(P,\setstr)$ such that $|\counts_{\Delta}(P,\setstr)-\approxcount_{\Delta}(P,\setstr)|\leq \alpha$.


\paragraph{Frequent Substring Mining} Let $\setstr=S_1,\dots, S_n$ be a database of documents in $\Sigma^{[1,\ell]}$. We define \substringmining($\setstr$, $\Delta$, $\thresh$) as follows. Compute a set $\mathcal{P}$ such that for any $P\in\Sigma^{[1,\ell]}$, we have $P \in \mathcal{P}$ if and only if $\counts_{\Delta}(P,\setstr)\geq\thresh$.
Similarly, we define  \qgrammining($\setstr$, $\Delta$, $\thresh$): compute a set $\mathcal{P}$ such that for any $P\in\Sigma^{q}$, we have $P \in \mathcal{P}$ if and only if $\counts_{\Delta}(P,\setstr)\geq\thresh$. That is, in \qgrammining\ we are only interested in substrings of length exactly $q$.

To obtain a differentially private \substringmining\ algorithm, we consider the approximate version of the problems.

    \begin{definition}
    [\approxmining($\setstr$, $\Delta$, $\thresh$)]\label{def:approxmining}
    Given a database $\setstr=S_1,\dots, S_n$ of documents from $\Sigma^{[1,\ell]}$, compute a set $\mathcal{P}$ such that for any $P\in\Sigma^{[1,\ell]}$, (1) if $\counts_{\Delta}(P,\setstr)\geq\thresh+\alpha$, then $P\in\mathcal{P}$; and
        (2) if $\counts_{\Delta}(P,\setstr)\leq\thresh - \alpha$, then $P\notin\mathcal{P}$.
    \end{definition}   

    The problem \approxqgram\ is defined in the same way as in Definition~\ref{def:approxmining}, adding the restriction $P\in\Sigma^{q}$. 
      For the problems \approxmining\ and \approxqgram, we also refer to $\alpha$ as the \emph{error}. 


 \subsubsection{Main results}

 In this paper, we give novel $\epsilon$-differentially private and $(\epsilon,\delta)$-differentially private data structures for $\counts_{\Delta}$. We obtain new solutions to frequent substring mining as a corollary. We stress that for all of our data structures, it is the \emph{construction algorithm} that satisfies differential privacy, so the resulting data structures can be queried ad libitum without further privacy loss. We first state our result for $\epsilon$-differential privacy, which gives with high probability both a data structure for \documentcount\ and for \substringcount\ with an additive error at most $O(\ell\cdot\mathrm{polylog}(n\ell+|\Sigma|))$. The data structure output by both Theorems~\ref{thm:main_eps} and~\ref{thm:main_epsdel} is a trie with $O(n\ell^2)$ nodes and a noisy count for each node. Patterns that are not present in the trie have a count that lies below the error bound with high probability.

 \begin{restatable}{theorem}{mainEpsThm}\label{thm:main_eps}
    Let $n,\ell\in\mathbb{N}$ 
    and $\Sigma$ an alphabet of size $|\Sigma|$. Let $\Delta\leq \ell$. For any $\epsilon>0$ and $0<\beta<1$, there exists an $\epsilon$-differentially private algorithm, which can process any database $\setstr=S_1,\dots,S_{n}$ of documents in $\Sigma^{[1,\ell]}$ and with probability at least $1-\beta$ outputs a data structure for $\counts_{\Delta}$ with additive error $O\left(\epsilon^{-1}\ell\log \ell(\log^2(n\ell/\beta)+\log |\Sigma|)\right)$. The data structure can be constructed in $O\left(n^2\ell^3\log\log(n\ell) +n^2\ell^4\right)$ time using $O(n^2\ell^4)$ space, stored in $O(n\ell^2)$ space, and answer queries in $O(|P|)$ time.
\end{restatable}
For $(\epsilon,\delta)$-differential privacy, we get better logarithmic factors in the error bound, and additionally improve the linear dependence on $\ell$ to $\sqrt{\ell\Delta}$.

\begin{restatable}{theorem}{mainEpsdelThm}\label{thm:main_epsdel}
    Let $n,\ell\in\mathbb{N}$ 
    and $\Sigma$ an alphabet of size $|\Sigma|$. Let $\Delta\leq \ell$. For any $\epsilon>0$, $\delta>0$ and $0<\beta<1$, there exists an $(\epsilon,\delta)$-differentially private algorithm, which can process any database $\setstr=S_1,\dots, S_{n}$ of documents in $\Sigma^{[1,\ell]}$ and with probability at least $1-\beta$ outputs a data structure for $\counts_{\Delta}$ with additive error $O\left(\epsilon^{-1}\sqrt{\ell\Delta\log(1/\delta)}\log \ell\left(\log (n\ell/\beta)+\sqrt{\log |\Sigma|\log\log \ell}\right)\right)$. The data structure can be constructed in $O(n^2\ell^4)$ time and space, stored in $O(n\ell^2)$ space, and answer queries in $O(|P|)$ time.
\end{restatable}

Given any threshold $\tau$, we can use our data structure to solve the $\approxmining$ problem with the same error guarantee: 
We simply traverse the trie and output every string corresponding to a node with a noisy count at least $\tau$. Note that the error in the noisy counts propagates to the threshold. Indeed, assuming the error of the noisy counts in the trie is at most $\alpha$, we can guarantee that the true count of any string we output is at least $\tau-\alpha$, and that we will output all strings with a true count of at least $\tau+\alpha$. However, strings whose true count lies within the interval $[\tau-\alpha,\tau+\alpha]$ may be output or not. Since this procedure only depends on the data structure, which was constructed privately, we can repeat this procedure for as many different thresholds as we want, without further privacy loss.

When we restrict to $q$-grams, the time and space complexity of the $\epsilon$-differential private construction algorithm improves by roughly a factor of $O(\ell)$, achieving the bounds of Theorem~\ref{thm:epsQgrams}. Notably, under $(\epsilon,\delta)$-differential privacy, the construction time and space complexity further drop to $O(n\ell( \log q +\log |\Sigma|))$ and $O(n\ell)$, respectively, as stated in  Theorem~\ref{thm:qgrams_eps_delta}.

\begin{restatable}{theorem}{thmEpsQgrams}\label{thm:epsQgrams}
 Let $n,\ell,q\in\mathbb{N}$,  $q\leq\ell$ 
 and $\Sigma$ an alphabet of size $|\Sigma|$. Let $\Delta\leq \ell$. For any $\epsilon>0$ and $0<\beta<1$, there exists an $\epsilon$-differentially private algorithm, which can process any database $\setstr=S_1,\dots,S_{n}$ of documents from $\Sigma^{[1,\ell]}$ and with probability at least $1-\beta$ outputs a data structure for computing $\counts_{\Delta}$ for all $q$-grams with additive error $O\left(\epsilon^{-1}\ell\log \ell(\log(n\ell/\beta)+\log |\Sigma|)\right)$. The data structure can be constructed in $O(n^2\ell^2\log q\log \log (n\ell)+n^2\ell^3)$ time using $O(n^2\ell^3)$ space, can be stored in $O(n \ell^2)$ space, and answers queries in $O(|P|)$ time.
\end{restatable}

\begin{restatable}{theorem}{thmQgramsEpsdelta}\label{thm:qgrams_eps_delta}
    Let $n,\ell,q\in\mathbb{N}$,  $q\leq\ell$ 
    and $\Sigma$ an alphabet of size $|\Sigma|$. Let $\Delta\leq \ell$. For any $\epsilon>0$, $\delta>0$, and $0<\beta<1$, there exists an $(\epsilon,\delta)$-differentially private algorithm, which can process any database $\setstr=S_1,\dots,S_{n}$ of documents from $\Sigma^{[1,\ell]}$ and with probability at least $1-\beta$ outputs a data structure for $\counts_{\Delta}$ for all $q$-grams with additive error $~O\left(\epsilon^{-1}\sqrt{\ell\Delta\log (n\ell)}\log q\left(\epsilon+\log\log q+ \log \frac{|\Sigma|}{\delta\beta}\right)\right)$. The data structure can be constructed in $O(n\ell( \log q +\log |\Sigma|))$ time and $O(n\ell)$ space.
\end{restatable}
 
Further down, we prove lower bounds stating that the error bounds from Theorems~\ref{thm:main_eps}-\ref{thm:qgrams_eps_delta} are tight up to polylogarithmic factors.

\subsubsection{Lower Bounds}
We note that since \documentcount\ can, in particular, be formulated as a counting query which can distinguish all elements from $\Sigma^{\ell}$, an $\Omega(\log(|\Sigma|^{\ell}))=\Omega(\ell\log |\Sigma|)$ lower bound for $\epsilon$-differentially private algorithms follows from well-established lower bounds (see e.g. the lecture notes by Vadhan~\cite{DBLP:books/sp/17/Vadhan17}; we discuss this in more detail in Section~\ref{sec:lower_bound}). We extend this lower bound in two ways.
(1) We show that it holds even if we only want to mine the patterns with a count approximately above a given threshold $\thresh$; (2) We show that it holds even if we restrict the output to patterns of a fixed length $m$, for any $m\geq 2\log \ell$. 

\begin{restatable}{theorem}{thmLB}\label{thm:lower_bound}
    Let $n,\ell\in \mathbb{N}$  and $\Sigma$ an alphabet of size $|\Sigma|\geq 4$. Let $m\geq 2\lceil\log \ell\rceil$. Let $\Delta\leq \ell$ and $\epsilon>0$. Let $\alg$ be an $\epsilon$-differentially private algorithm, which takes as input any database $\setstr=S_1,\dots, S_n$ of documents in $\Sigma^{\ell}$ and a threshold $\thresh$.  
    If $\alg$ computes with probability at least $2/3$ a set $\mathcal{P}$ such that for any $P\in\Sigma^{m}$,
   \begin{enumerate}
       \item if $\counts_{\Delta}(P,\setstr)\geq\thresh+\alpha$, then $P\in\mathcal{P}$ and
       \item if $\counts_{\Delta}(P,\setstr)\leq\thresh - \alpha$, then $P\notin\mathcal{P}$,
   \end{enumerate}
    then $\alpha=\Omega(\min(n,\epsilon^{-1}\ell\log |\Sigma|))$.
\end{restatable}

Theorems~\ref{thm:LB_approx_substring_count} and~\ref{thm:LB_approx_document_count} prove that the error bounds of our $(\epsilon,\delta)$-differentially private data structures for \substringcount\ and \documentcount\ are also essentially optimal. Additionally, a corollary of Theorem~\ref{thm:LB_approx_substring_count} that we give in Section~\ref{sec:lower_bound} states the same lower bound on the error also for the simpler \approxqgram\ problem. Theorem~\ref{thm:LB_approx_substring_count} is proved via a direct example. Theorem~\ref{thm:LB_approx_document_count} is proved via reduction to the 1-way marginals problem.

\begin{restatable}{theorem}{LBepsdelSubstr}\label{thm:LB_approx_substring_count}
 Let $n,\ell\in \mathbb{N}$  and $\Sigma$ an alphabet of size $|\Sigma|\geq 2$. Let $\beta,\delta\in[0,1)$.
Every $(\epsilon,\delta)$-differentially private algorithm for \substringcount\ which takes as an input any database of $n$ documents from $\Sigma^{\ell}$ and is 
$\alpha$-accurate with probability at least $1-\beta$ 
must satisfy either        $$\alpha=\Omega(\ell)$$
        or 
        \begin{equation}\label{eq:approx_LB}
        \epsilon\geq \ln\left(\frac{1-\beta-\delta}{\beta}\right).
        \end{equation}
\end{restatable}

\begin{restatable}{theorem}{LBepsdelDocument}\label{thm:LB_approx_document_count}
Let $n,\ell\in \mathbb{N}$  and $\Sigma$ an alphabet of size $|\Sigma|=s\geq 3$. Let $\epsilon\in(0,1]$ and $\delta\in[0,1)$. Every $(\epsilon,\delta)$-differentially private algorithm for \documentcount\ which takes as an input any database of $n$ documents from $\Sigma^{\ell}$ and is
$\alpha$-accurate with probability at least $2/3$  must satisfy:       
\begin{itemize}
    \item if $\delta>0$ and $\delta=o\left(\frac{1}{n}\right)$, then $\alpha
    =\Omega\left(\min\left(\frac{\sqrt{\ell}}{\epsilon\log_{s-1}\ell\log \ell}, \frac{\sqrt{\ell}}{\epsilon\log \ell}\right)\right)$
    \item if $\delta=0$, then $\alpha=\Omega\left(\min\left(\frac{\ell}{\epsilon\log_{s-1}\ell},\frac{\ell}{\epsilon}\right)\right)$.
\end{itemize}
\end{restatable}

 \subsubsection{Counting on Trees}
On the way to proving our main result, we develop a new technique to differentially privately compute count functions on a tree, which may be of independent interest. We motivate this using two examples. First, let $T$ be a tree, where every leaf corresponds to an element in the universe. The count of every leaf is the number of times the element appears in a given data set, and the count of a node is the sum of the counts of the leaves below. This problem captures any hierarchical composition of data items (i.e., by zip code, area, state), and has been studied under differential privacy (e.g. \cite{DBLP:conf/sigmod/ZhangXX16,DBLP:conf/icalp/GhaziK0M023}). In particular, as noted in \cite{DBLP:conf/icalp/GhaziK0M023}, it can be solved via a reduction to differentially private range counting over the leaf counts. For $\epsilon$-differential privacy, the binary tree mechanism by Dwork~et~al.~\cite{Dwork2010} gives an error of roughly $O(\log^2 u)$ for this problem, where $u$ is the size of the universe. For $(\epsilon,\delta)$-differential privacy, the range counting problem has been extensively studied recently~\cite{DBLP:conf/stoc/Cohen0NSS23,DBLP:conf/colt/KaplanLMNS20,DBLP:conf/stoc/AlonLMM19,DBLP:conf/stoc/BunDRS18,DBLP:conf/focs/BunNSV15,DBLP:journals/toc/BeimelNS16}. As a second example, let again $T$ be a tree and the leaves the items of the universe, however, we additionally have a \emph{color} for every item of the universe. Now, the goal is to compute for every node the \emph{number of distinct colors} of elements present in the data set corresponding to leaves below the node. We call this problem the \emph{colored tree counting problem}. This is similar to the \emph{colored range counting problem}, which is well-studied in the non-private setting \cite{KaplanRSV08, Gao021, Gao23}, and was also recently studied with differential privacy under the name of ``counting distinct elements in a time window"~\cite{DBLP:conf/innovations/Ghazi0NM23}. 
We give an $\epsilon$-differentially private algorithm that can solve 
both problems with error $O(\log^2 u \log h)$, where $h$ is the height of the tree. In general, our algorithm can compute any counting function that is (i) monotone, in the sense that any node's count cannot be larger than the sum of the counts of its children, and (ii) has bounded \emph{$L_1$-sensitivity} on the leaves, i.e., the true counts of the leaves do not differ too much on neighboring data sets (see Definition~\ref{def:sensitivity} for a formal definition of sensitivity). 
We next state the lemma in full generality.

\begin{restatable}{theorem}{thmTreeCount}\label{thm:tree_count}
    Let $\chi$ be a universe and $T=(V,E)$ a tree height $h$. Let $L=l_1,\dots,l_k\subseteq V$ be the set of leaves in $T$. Let $c:V\times \chi^{*}\rightarrow \mathbb{N}$ be a function, which takes as input a node $v$ from $T$ and a data set $\setstr$ from $\chi^{*}$, and computes a count with the following properties:
    \begin{itemize}
        \item $c(v,\setstr)\leq\sum_{u \textnormal{ is a child of }v} c(u,\setstr)$ for all $v\in V\setminus L$ and $\setstr\in \chi^{*}$;
        \item  $\sum_{i=1}^k |c(l_i,\setstr)-c(l_i,\setstr')|\leq d$, for all neighboring $\setstr$ and $\setstr'$ and some $d\in\mathbb{N}$.
    \end{itemize}
    Then for any $\epsilon>0$ and $0<\beta<1$, there exists an $\epsilon$-differentially private algorithm computing $\hat{c}(v)$ for all nodes in $T$ such that $\max_{v\in V}|\hat{c}(v,\setstr)-c(v,\setstr)|=O(\epsilon^{-1}d\log |V|\log h \log(hk/\beta))$ with probability at least $1-\beta$. In particular, if $d=1$, then the maximum error is at most $O(\epsilon^{-1}\log |V|\log h \log(hk/\beta))$ with probability at least $1-\beta$.

\end{restatable}
We also show an improvement for $(\epsilon,\delta)$-differential privacy, if additionally the sensitivity of $c(v,\cdot)$ is small for each node $v$.

\begin{restatable}{theorem}{thmEpsdelTrees}\label{thm:tree_count_epsdel}
    Let $\chi$ be a universe and $T=(V,E)$ a tree height $h$.  Let $L=l_1,\dots,l_k\subseteq V$ be the set of leaves in $T$. Let $c:V\times \chi^{*}\rightarrow \mathbb{N}$ a function, which takes as input a node $v$ from $T$ and a data set $\setstr$ from $\chi^{*}$, and computes a count with the following properties:
    \begin{itemize}
        \item $c(v,\setstr)\leq\sum_{u \textnormal{ is a child of }v} c(u,\setstr)$ for all $v\in V\setminus L$ and $\setstr\in \chi^{*}$;
        \item $\sum_{i=1}^k |c(l_i,\setstr)-c(l_i,\setstr')|\leq d$, for all neighboring $\setstr$ and $\setstr'$ and some $d\in\mathbb{N}$.
        \item $|c(v,\setstr)-c(v,\setstr')|\leq \Delta$, for all $v\in T$ and all neighboring $\setstr$ and $\setstr'$.
    \end{itemize}
    Then for any $\epsilon>0$, $\delta>0$, and $0<\beta<1$, there exists an $(\epsilon,\delta)$-differentially private algorithm computing $\hat{c}(v)$ for all nodes in $T$ such that $$\max_{v\in V}|\hat{c}(v,\setstr)-c(v,\setstr)|=O(\epsilon^{-1}\sqrt{d\Delta\log |V|\log(1/\delta)\log(hk/\beta)}\log h )$$ with probability at least $1-\beta$. 

\end{restatable}

\subsection{Technical Overview}
{\bf Simple approach.}  
Before we introduce our algorithm, we present a simple trie-based approach (a similar strategy was adopted in previous work, e.g. \cite{DBLP:conf/cikm/BonomiX13, DBLP:conf/sigmod/ZhangXX16,DBLP:conf/ccs/ChenAC12,DBLP:conf/kdd/ChenFDS12,DBLP:conf/nips/KimGKY21,DBLP:conf/securecomm/KhatriDH19}). For simplicity, we focus on \substringcount. A private trie for the database is constructed top-down. Starting at the root, we create a new child node for every letter in the alphabet, connected by an edge labeled with the corresponding letter. Then, for each leaf, we compute how many times the corresponding letter occurs in $\mathcal{D}$ and add noise to this count. If the noisy count exceeds a certain threshold, we expand this node by adding a child for every letter in the alphabet. We continue expanding the trie from the newly added leaves, but in later rounds, when computing counts for a leaf, instead of considering a single letter, we count how many times the string obtained by concatenating the labels of the edges on the root-to-leaf path occurs as a substring in $\mathcal{D}$.

We repeat this process until no noisy count exceeds the threshold. The main advantage of this top-down approach is to avoid considering the entire universe by excluding patterns that are not frequent as soon as a prefix is not frequent. To make the algorithm $\epsilon$-differentially private, the noise we add has to scale approximately with the $L_1$-sensitivity of the true counts. The $L_1$-sensitivity is here defined as how much the true counts of the nodes in the trie can differ in total if we replace a document in the database. Assume $\setstr'=\setstr\setminus\{S\}\cup \{S'\}$. Then, the count of a node is different on $\setstr$ and $\setstr'$ if and only if the string represented by the node is a substring of $S$ or $S'$ (or both). However, $S$ and $S'$ can each have up to $\Omega(\ell^2)$ different substrings. Thus, this approach yields an error of $\Omega(\ell^2)$.

\begin{figure}
    \centering
    \includegraphics[width=.48\columnwidth]{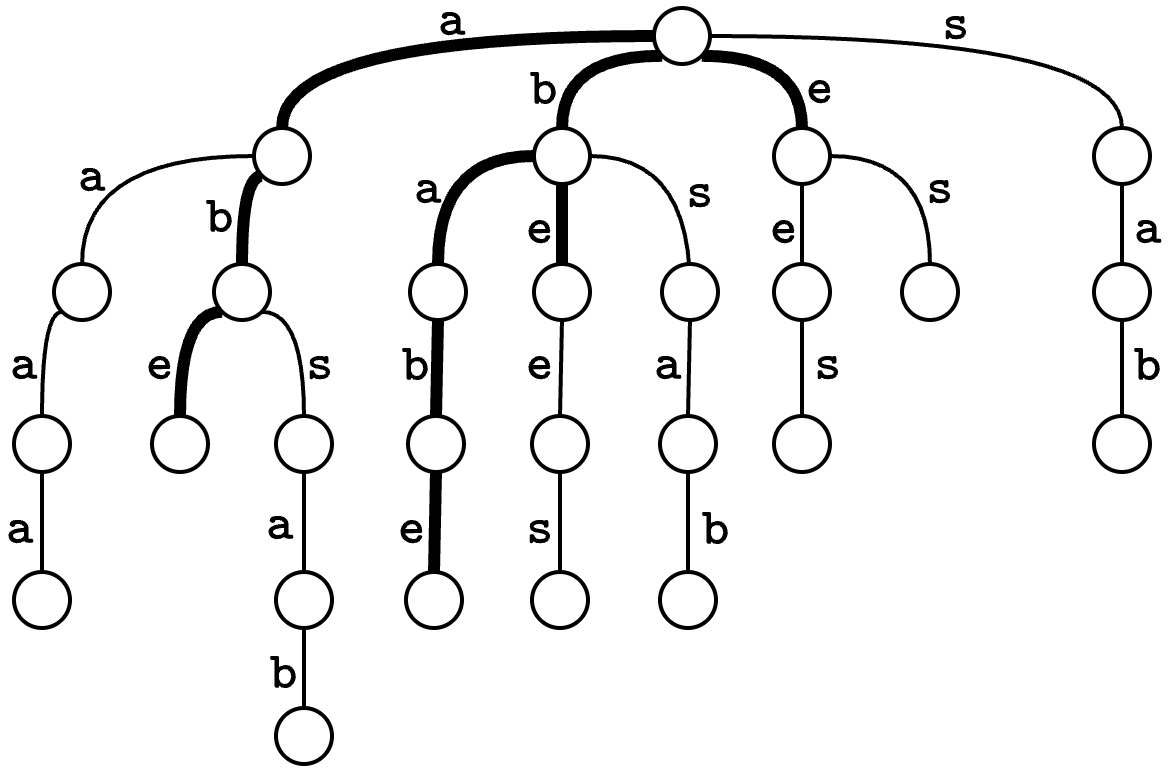}
    \caption{The trie of all suffixes of the strings in database $\setstr$ from Example~\ref{ex:doc_substr_count}. The paths corresponding to all suffixes of the string $\texttt{babe}$ are in bold.}
    \label{fig:trie_paths}
\end{figure}

{\bf Idea for improvement.} 
An important insight is that while there are $\Omega(\ell^2)$ nodes in the trie that can be influenced by $S$ and $S'$, these nodes can lie on at most $2\ell$ different root-to-leaf paths, namely, the paths corresponding to the suffixes of $S$ and $S'$. See Figure~\ref{fig:trie_paths}. Our idea is to leverage this property. For this, we make use of the \emph{heavy path decomposition}~\cite{SLEATOR1983362} of a tree $T$, which partitions $T$ into disjoint \emph{heavy paths} and has the property that any root-to-leaf path in $T$ crosses at most a logarithmic number (in the size of $T$) of such heavy paths. We will use these heavy paths to add noise in a way such that the error depends on the number of heavy paths containing nodes with a count that is different on $\setstr$ and $\setstr'$.

However, there are two main challenges.  We cannot construct the heavy path decomposition during the top-down construction above, since this requires knowing the size of the subtrees before they have been explored, but at the same time, we need the noisy counts to decide which nodes to expand further.  
Alternatively, we could first construct the full trie of the universe (all possible substrings of length $\ell$) and then compute the heavy path decomposition. But in this case, the size of the trie would be $|\Sigma|^{\ell}$, so the number of heavy paths on a root-to-leaf path would only be bounded by $\ell\log\Sigma$, and therefore, the $\ell$ root-to-leaf paths which contain nodes of different counts in $\setstr$ and $\setstr'$ can cross up to $\Omega(\ell^2\log \Sigma)$ heavy paths in total.


{\bf Our algorithm.} 
We now give an overview of our full algorithm. As a first step, we reduce the universe size from $|\Sigma|^{\ell}$ to $n^2\ell^3$, by computing a set of \emph{candidate frequent strings} in a differentially private way. 
Recall that we need substrings not occurring in our database $\mathcal{D}$ to appear in our candidate set with non-zero probability. We show how to efficiently obtain this without exhaustively considering all such substrings.

We iteratively construct a collection $\mathcal{P} =\mathcal{P}_{2^0}, \ldots, \mathcal{P}_{2^{\lfloor \log \ell\rfloor}}$ of sets of strings, such that the strings in $\mathcal{P}_{2^k}$ have length $2^k$, as follows.  We first construct $\mathcal{P}_{2^0}$ by computing noisy counts for all letters in the alphabet, and keep only those with a noisy count above a certain threshold. To construct $\mathcal{P}_{2^k}$, we compute noisy counts of all strings obtained by concatenating two strings from $\mathcal{P}_{2^{k-1}}$, again keeping only those with a noisy count above the threshold.  We show that the collection only contains substrings from $\mathcal{D}$ with high probability. 
Note that a document $S$ 
can have at most $\ell$ substrings of any fixed length, and we consider a logarithmic number of lengths. Therefore, our noise needs to scale with $\ell\log \ell$. We then use the following observation: any string of a given length $m$ which is not a power of two has an overlapping prefix and suffix whose length is a power of two. Our final candidate set consists of the collection $\mathcal{P}$ together with all strings covered by a 
prefix and suffix from $\mathcal{P}$.

Next, we build a trie from our candidate set and decompose it into heavy paths. For every heavy path, we compute (i) a differentially private noisy count of its root, and (ii) a differentially private estimate of the \emph{prefix sums of the difference sequence of the counts} on the path going down. Specifically, for every node on the heavy path, we get an approximation of the \emph{difference} between its count and the count of the root. We can compute (ii) with low noise using a generalized variant of the binary tree mechanism~\cite{Dwork2010}. We show that since any root-to-leaf path can cross at most $\log (n^2\ell^3)=O(\log (n\ell))$
heavy paths, and a document $S$ 
can only influence counts on the at most $\ell$ paths corresponding to its suffixes, this can be done with an error which is $O(\ell\cdot \mathrm{polylog}(n\ell))$. This strategy gives our main result of Theorem~\ref{thm:main_eps}.

The general tree counting lemma (Theorem~\ref{thm:tree_count}) is obtained by building a heavy path decomposition on the given tree, and computing (i) and (ii).
The improvement for $(\epsilon,\delta)$-differential privacy comes from the fact that if the $L_1$-norm of a vector is bounded by $L$, and the $L_{\infty}$-norm is bounded by $\Delta$, then by Hölder's inequality, the $L_{2}$-norm is bounded by $\sqrt{L\Delta}$. We use this fact several times in the analysis to improve the linear dependency on $\ell$ to a dependency on $\sqrt{\ell\Delta}$.

The improvement for $q$-grams stems from two facts: 1) since we only want to compute the counts for patterns of a fixed length, we do not need to compute the full set of candidate strings and can significantly simplify the algorithm, and 2) for $(\epsilon,\delta)$-differential privacy, we show that we can avoid computing noisy counts for patterns whose true count is 0, without violating privacy.


\subsection{Related Work}
{\bf Private sequential pattern mining and $q$-gram extraction.} As mentioned above, two very related well-studied problems are those of \emph{frequent sequential pattern mining}~\cite{DBLP:conf/cikm/BonomiX13,DBLP:conf/sigmod/ZhangXX16,DBLP:conf/kdd/ChenFDS12,app12042131,DBLP:conf/securecomm/KhatriDH19} and \emph{q-gram extraction}~\cite{DBLP:conf/nips/KimGKY21,DBLP:conf/ccs/ChenAC12}. We also mentioned that these works do not provide theoretical guarantees on the error. We now describe some of the relevant solutions in more detail and argue why our error bounds yield an improvement over previous approaches. In the following, by ``error" we either mean the additive error on the counts, or the value of $\alpha$ in the definition of $\approxmining$ and $\approxqgram$, depending on which problem the paper is considering.

Several papers use a variant of the trie-based approach described in the technical overview~\cite{DBLP:conf/cikm/BonomiX13, DBLP:conf/sigmod/ZhangXX16,DBLP:conf/ccs/ChenAC12,DBLP:conf/kdd/ChenFDS12,DBLP:conf/nips/KimGKY21,DBLP:conf/securecomm/KhatriDH19}. 
Some of them~\cite{DBLP:conf/cikm/BonomiX13,DBLP:conf/kdd/ChenFDS12} use a \emph{prefix trie}, i.e., they build a trie of strings that appear frequently as a prefix of a string in the data set. While this can be used for estimating frequent elements for large universes, it does not solve the problem of counting $q$-grams or substrings, since it is not possible to reconstruct frequent substrings from frequent prefixes (e.g. there could be a frequent substring that is not the suffix of any frequent prefix).

Khatri~et~al.~\cite{DBLP:conf/securecomm/KhatriDH19} also construct a tree in an online manner (however, appending new letters \emph{in front} instead of at the end), and compute the count of how often a string in the trie appears as a \emph{suffix} as a counting measure, which leads to a similar issue when trying to estimate substring counts. Zhang~et~al.~\cite{DBLP:conf/sigmod/ZhangXX16} consider the problem of computing tree counts, where every leaf is an element of the universe, and the count of every node is the number of leaves below which are in the data set. As a first step, they build a pruned tree such that every element in the tree has at least a certain count. As a second step, they compute a differentially private count of every leaf. The count of any node is computed as the sum of the counts of the leaves below. However, when computing the count of a node from the noisy counts of the leaves, the noise values sum up, so this can yield large errors for internal nodes. 

Chen~et~al.~\cite{DBLP:conf/ccs/ChenAC12} and Kim~et~al.~\cite{DBLP:conf/nips/KimGKY21} consider the problem of computing frequent $q$-grams and use a trie-based approach very similar to the simple approach described in the technical overview. While they use different heuristics and parameter tuning to improve accuracy in practice, the worst-case theoretical bound on the error remains $\Omega(\ell^2)$. 
We note that  Kim~et~al.~\cite{DBLP:conf/nips/KimGKY21}, when computing the frequent $q$-grams, use a candidate set consisting of any frequent $(q-1)$-gram appended by any frequent $1$-gram, and its intersection with the set consisting of any frequent $1$-gram appended by any frequent $(q-1)$-gram. We use a similar idea in the pruning process in the first step of our algorithm, except that we \emph{double} the length of the $q$-grams at every step. 

Prado~et~al.~\cite{app12042131} use a graph-based approach for finding frequent patterns. In their work, a bi-partite graph is produced linking users to (frequent) patterns that appear in the user's document; this graph is then privatized using an edge-differentially private algorithm. Note that this algorithm satisfies a weaker privacy guarantee, as one user string can be connected to several patterns. 

{\bf Further related work.} Similar problems which have been extensively studied are frequent itemset mining \cite{DBLP:journals/pvldb/ZengNC12, DBLP:journals/compsec/ChengSXL15, DBLP:conf/sp/WangLJ18,DBLP:journals/access/XiongCHTHCQ18,DBLP:conf/wise/DingCJ17}, where the data set consists of a set of items for each person and the goal is to find frequent subsets, and frequent sequence mining~\cite{DBLP:conf/icde/XuSCLX15, Liang2023,DBLP:journals/compsec/ChengSXTL15,DBLP:conf/dasfaa/LiWYCYL18,DBLP:conf/icic/ZhouL18,DBLP:conf/ideas/LeKAY19}, where the goal is to find frequent \emph{subsequences}, instead of frequent \emph{substrings}. Both of these problems are different from those considered in this work since they do not require a pattern to occur as a consecutive substring. Another related line of work is differentially private pattern counting on a \emph{single} string with different definitions of neighboring (and therefore different privacy guarantees)~\cite{DBLP:conf/icdm/ChenDFLPS21, DBLP:conf/innovations/Steiner24,DBLP:conf/icml/FichtenbergerHU23}. The problems of frequent string mining and frequent sequence mining have also been studied in the local model of differential privacy~\cite{DBLP:conf/icde/WangXYHSS018,DBLP:journals/jksucis/WangH22,DBLP:conf/satml/ChadhaCDFHJMT24}.

%
\subsection{Outline}
In Section~\ref{sec:prelims}, we recall some basic definitions and results on differential privacy and strings. In Section~\ref{sec:puredp}, we give our data structures for $\epsilon$-differential privacy, proving Theorem~\ref{thm:main_eps} and Theorem~\ref{thm:epsQgrams}. In Section~\ref{sec:epsdel}, we give our data structures for $(\epsilon,\delta)$-differential privacy, proving Theorem~\ref{thm:main_epsdel} and Theorem~\ref{thm:qgrams_eps_delta}. In Section~\ref{sec:tree_count}, we give our general tree counting algorithm and prove Theorems~\ref{thm:tree_count} and \ref{thm:tree_count_epsdel}. Finally, in Section~\ref{sec:lower_bound}, we prove our lower bounds.

\section{Preliminaries}\label{sec:prelims}
In this work, we use $\log$ to denote the binary logarithm and $\ln$ to denote the natural logarithm. We use $[a,b]$ to denote the interval of integers $\{a,a+1,a+2,\dots,b-1,b\}$.

\subsection{String Preliminaries.}
A string~$S$ of length~$|S|=\ell$ is a sequence $S[0]\cdots S[\ell-1]$ of $\ell$ characters drawn from an alphabet $\Sigma$ of  size $|\Sigma|$. The string $S[i]\cdots S[j]$, 
denoted $\substring{S}{i}{j}$, is called a \emph{substring} of $S$; $\substring{S}{0}{j-1}$ and $\substring{S}{i}{\ell-1}$  are called the $j^{th}$ \emph{prefix} and  $i^{th}$ \emph{suffix} of $S$, respectively.  Let $P$ and $S$ be strings over an alphabet $\Sigma$. We say that $P$ \emph{occurs} in $S$ iff there exists an $i$ such that $\substring{S}{i}{i+ |P|-1} = P$. We use $\Sigma^{[a,b]}$ to denote all strings $S$ which satisfy $a\leq |S|\leq b$. The concatenation of two strings $A,B$ is defined as $A\cdot B=A[0]\cdots A[|A|-1]B[0]\cdots B[|B|-1]$ and is sometimes simply denoted by $AB$. Given any two sets of strings $\mathcal{A},\mathcal{B}$, we define $\mathcal{A}\circ\mathcal{B}=\{A\cdot B\mid A\in\mathcal{A}, B\in\mathcal{B}\}$. 

A \emph{trie} for a collection of strings $\mathcal{C}=S_1, \ldots, S_n$, denoted $T_\mathcal{C}$, is a rooted labeled tree, such that: (1) The label on each edge is a character of one or more $S_i$. (2) Each string in $\mathcal{C}$ is represented by a path in $T_\mathcal{C}$ going from the root down to some node (obtained by concatenating the labels on the edges of the path).
(3) Each root-to-leaf path represents a string from $\mathcal{C}$.
(4) Common prefixes of two strings share the same path maximally.
For a node $v\ \in T_\mathcal{C}$, we let $\str{v}$ denote the string obtained from concatenating the labels on the path's edges from the root to $v$. 

A \emph{compacted trie} is obtained from $T_\mathcal{C}$ by dissolving all nodes except the root, the branching nodes, and the leaves, and concatenating the labels on the edges incident to dissolved nodes to obtain \emph{string} labels for the remaining edges. 
The \emph{string depth} $\sd(v)=|\str{v}|$ of any branching node $v$ is the total length of the strings labeling the path from the root to $v$. The \emph{frequency} $f(v)$ of node $v$ is the number of leaves in the subtree rooted at $v$. Throughout the paper, we assume that $f(v)$ and $\sd(v)$ values are stored in each branching node $v$. This can be obtained in $O(|\mathcal{C}|)$ time with a single trie traversal using $O(|\mathcal{C}|)$ total extra space. The compacted trie can be represented in $O(|\mathcal{C}|)$ space and computed in $O(|\mathcal{C}|)$ time~\cite{DBLP:conf/cpm/KasaiLAAP01}.

Let $S$ be a string over an alphabet $\Sigma$. The \emph{suffix tree} of a string $S$, denoted by $\ST(S)$, is the compacted trie of the set of all suffixes of $S$. Each leaf of $\ST(S)$ is identified by the starting position in $S$ of the suffix it represents. We can construct $\ST(S)$ in $O(\textrm{sort}(|S|, |\Sigma|)$, where $\textrm{sort}(x, y)$ is the time to sort $|S|$ integers from an universe of size $|\Sigma|$~\cite{DBLP:conf/focs/Farach97, FCFM2000}. For instance, for any polynomial-size alphabet $\Sigma=[0,|S|^{O(1)}]$, we can construct $\ST(S)$ in $O(|S|)$ time.


\subsection{Privacy Preliminaries}
We collect some important definitions and results here. 
First, we define the notion of sensitivity.
   \begin{definition}[$L_p$-sensitivity]\label{def:sensitivity}
     Let $f$ be a function $f:\chi^{*}\rightarrow \mathbb{R}^k$ for some universe $\chi$. The \emph{$L_p$-sensitivity of $f$} is defined as
    $
    \max_{\setstr\sim\setstr'}||f(\setstr)-f(\setstr')||_{p}$. 
    \end{definition}

Next, we recall simple composition of differentially private algorithms.
\begin{lemma}[Simple Composition~\cite{DBLP:conf/stoc/DworkL09}]\label{lem:composition_theorem} Let $A_1$ be an $(\epsilon_1,\delta_1)$-differentially private algorithm $\chi^{*}\rightarrow \mathrm{range}(A_1)$ and $A_2$ an $(\epsilon_2,\delta_2)$-differentially private algorithm $\chi^{*}\times\mathrm{range}(A_1)\rightarrow \mathrm{range}(A_2)$. Then $A_1\circ A_2$ is $(\epsilon_1+\epsilon_2,\delta_1+\delta_2)$-differentially private. \end{lemma}
%
\subsubsection{Laplace mechanism }
\begin{definition}\label{def:laplace}
The \emph{Laplace distribution} centered at $0$ with scale $b$ is the distribution with probability density function 
$
    f_{\Lap(b)}(x)=\frac{1}{2b}\exp\left(\frac{-|x|}{b}\right)$.
 We use $Y\approx \Lap(b)$ or just $\Lap(b)$ to denote a random variable $Y$ distributed according to $f_{\Lap(b)}(x)$.
\end{definition}
The Laplace distribution satisfies the following tailbound:

\begin{lemma}[Laplace tailbound]\label{lem:lap_tail}
    If $Y\approx \Lap(b)$, then 
    $\Pr[|Y|\geq tb]=e^{-t}$.
\end{lemma}

Next, we define the Laplace mechanism.
\begin{lemma}[Laplace Mechanism~\cite{Dwork2006}] \label{lem:Laplacemech} Let $f$ be any function $f:\chi^{*}\rightarrow \mathbb{R}^k$ with $L_1$-sensitivity $\Delta_1$. Let $Y_i\approx \Lap(\Delta_1/\epsilon)$ for $i\in[k]$. The mechanism is defined as
$A(\setstr)=f(\setstr)+(Y_1,\dots,Y_k)$
satisfies $\epsilon$-differential privacy.
\end{lemma}
Lemma~\ref{lem:lap_tail} and Lemma~\ref{lem:Laplacemech} give the following corollary.
\begin{corollary}\label{cor:Laplace_mech}
     Let $f$ be a function $f:\chi^{*}\rightarrow \mathbb{R}^k$ for some universe $\chi$ with $L_1$-sensitvity at most $\Delta_1$. Then there exists an $\epsilon$-differentially private algorithm $A$ which for any $\setstr\in \chi^{*}$ outputs $A(\setstr)$ satisfying $||A(\setstr)-f(\setstr)||_{\infty}\leq \epsilon^{-1}\Delta_1\ln(k/\beta)$ with probability at least $1-\beta$.
\end{corollary}

\subsubsection{Gaussian Mechanism}

\begin{definition}[Normal Distribution] The \emph{normal distribution} centered at $0$ with variance $\sigma^2$ is the distribution with the probability density function
\begin{align*}
f_{N(0,\sigma^2)}(x)=\frac{1}{\sigma\sqrt{2\pi}}\exp\left(-\frac{x^2}{2\sigma^2}\right)
\end{align*}

\end{definition}
 We use $Y\approx N(0,\sigma^2)$ or sometimes just $N(0,\sigma^2)$ to denote a random variable $Y$ distributed according to $f_{N(0,\sigma^2)}$.

The normal distribution satisfies the following tailbound.
\begin{lemma}[Gaussian Tailbound]\label{lem:gaussian_tail}
Let $Y\approx N(\mu,\sigma^2)$. Then 
    $\Pr[|Y-\mu|\geq t]\leq 2e^{-\frac{t^2}{2\sigma^2}}$.

\end{lemma}

 Next, we define the Gaussian mechanism.
 
\begin{lemma}[Gaussian mechanism~\cite{BlumDMN05}]\label{lem:gaussianmech}
 Let $f$ be any function $f:\chi^{*}\rightarrow \mathbb{R}^k$ with $L_2$-sensitivity $\Delta_2$.
Let $\epsilon\in(0,1)$, $c^2>2\ln(1.25/\delta)$, and $\sigma\geq c\Delta_2(f)/\epsilon$. Let $Y_i\approx N(0,\sigma^2)$ for $i\in[k]$. Then the mechanism defined as $
A(\setstr)=f(\setstr)+(Y_1,\dots,Y_k)$
satisfies $(\epsilon,\delta)$-differential privacy.
\end{lemma}

Lemma~\ref{lem:gaussian_tail} and Lemma~\ref{lem:gaussianmech} give the following corollary.
\begin{corollary}\label{cor:gaussian_mech}
     Let $f$ be a function $f:\chi^{*}\rightarrow \mathbb{R}^k$ for some universe $\chi$ with $L_2$-sensitvity at most $\Delta_2$. Then there exists an $(\epsilon,\delta)$-differentially private algorithm $A$ which for any $\setstr\in \chi^{*}$ outputs $A(\setstr)$ satisfying $||A(\setstr)-f(\setstr)||_{\infty}\leq 2\epsilon^{-1}\Delta_2 \sqrt{\ln(2/\delta)\ln(2k/\beta)}$ with probability at least $1-\beta$.    
\end{corollary}



\section{Counting with Pure Differential Privacy}\label{sec:puredp}
In this section, we show how to construct our data structure for Theorem~\ref{thm:main_eps}.
\mainEpsThm*

For simplicity, we prove the theorem for $\Delta=\ell$. Since $\counts_{\Delta}(P,S)\leq \counts(P,S)$ holds for all $\Delta\leq \ell$ and all strings $S\in\Sigma^{[1,\ell]}$, the same proof also works for $\Delta<\ell$. 

We make the following key observation, which we will use repeatedly.
\begin{observation}\label{lem:pattern_counts}
   For any document $S\in\setstr$ and any length $m\leq \ell$, the cumulative count $\sum_{P\in \Sigma^m}\counts(P,S)$ of all length-$m$ substrings of $S$ is at most $\ell$.
\end{observation}
This observation holds because, since the length of $S$ is bounded by $\ell$, $S$ has at most $\ell-m+1\leq \ell$ fragments of length $m$ (one for each possible starting position). Thus the total number of occurrences of its substrings of length $m$ is bounded by the same quantity. Observation~\ref{lem:pattern_counts} implies the following.
\begin{corollary}\label{cor:l1sens_same_length}
     The $L_1$-sensitivity of $(\counts(P,\setstr))_{P\in \Sigma^m}$ is bounded by $2\ell$, for any $m\leq \ell$. 
\end{corollary}

%

\paragraph{Algorithm Overview} Our construction algorithm has six main steps.

 \textit{Step 1.} We run an $\epsilon$-differentially private algorithm to compute a set $\candidates$ of \emph{candidate frequent} strings 
 such that the true count in $\setstr$ of any string not included in $\candidates$ is small, and the set $\candidates$ is not too large. 

  \textit{Step 2.} We build the trie $\candidatetrie$ of the candidate set $\candidates$, and compute its \emph{heavy path decomposition}.  

  \textit{Step 3.} For every root $r$ of a heavy path, we compute a differentially private estimate of $\counts(\str{r},\setstr)$. 

 \textit{Step 4.}
  For every heavy path $p=v_0,v_1,\dots, v_{|p|-1}$ with root $r=v_0$, we consider the \emph{difference sequence} of counts along the path: i.e., for a node $v_i$ on the path and its parent $v_{i-1}$, the $i$th element in the difference sequence is given by $\counts(\str{v_i},\setstr)-\counts(\str{v_{i-1}},\setstr)$, for $i=1,\dots, |p|-1$. We use the binary tree mechanism \cite{Dwork2010} to compute a differentially private estimate of all \emph{prefix sums}, $\mathrm{sums}^*_p$, of each of the difference sequences. 
  
  \textit{Step 5.} We compute the noisy count of every node $v$ in $\candidatetrie$, thus of every substring in $\candidates$: For every node $v_i$ on a heavy path $v_0,v_1,\dots, v_{|p|-1}$,  we set $\approxcount(\str{v_i},\setstr) = \approxcount(\str{v_0},\setstr) + \mathrm{sums}^*_p[i]$.
    
  \textit{Step 6.}
  We finally prune $\candidatetrie$ removing the subtrees with sufficiently small noisy counts, and return the pruned trie as our $\epsilon$-differentially private data structure.
\\\\
In the rest of the section, we give the details of the algorithm.

 \subsection{Step 1: Computing a Candidate Set}\label{sec:construct_P_epsilon_DP}

As a first step, we show how to compute the candidate set $\candidates$ while satisfying $\epsilon$-differential privacy by doing the following.


\begin{enumerate}
    \item[1.1] Inductively construct sets $\candidatespruned_{2^0},\dots, \candidatespruned_{2^j}$, with $j=\lfloor \log \ell\rfloor$, where $\candidatespruned_{2^k}$ contains only strings of length $2^k$ which appear sufficiently often as substrings in $\setstr$ according to noisy counts. Let $\epsilon_1=\epsilon/(\lfloor\log \ell\rfloor +1)$. For every $k=1,\ldots, \lfloor \log \ell\rfloor$, we use an $\epsilon_1$-differentially private algorithm to construct $\candidatespruned_{2^k}$ from $\candidatespruned_{2^{k-1}}$, such that the full algorithm fulfills $\epsilon$-differential privacy. 
    \item[1.2] 
    Compute, for each $m\leq \ell$ which is not a power of $2$, the set $\candidates_m$ of all strings $P$ of length $m$ whose length-$2^k$ prefix and suffix are in $\candidatespruned_{2^k}$, 
    where $k=\lfloor\log m\rfloor$; 
    \item[1.3] The set $\candidates$ is finally obtained as the union of all the sets computed in Steps 1.1 and 1.2.
\end{enumerate}
More precisely, the algorithm works as follows. 

\paragraph{Computing $\candidatespruned_{2^0}$.} First, we compute a noisy count for all letters $\gamma\in\Sigma$. By Corollary~\ref{cor:l1sens_same_length}, the sensitivity of $(\counts(\gamma,\setstr))_{\gamma\in \Sigma}$ is bounded by $2\ell$.  
Let $\beta_1=\beta/(\lfloor\log \ell\rfloor +1)$.
Using the algorithm given by Corollary~\ref{cor:Laplace_mech}, we compute an estimate $\approxcount(\gamma,\setstr)$ such that  $$\max_{\gamma\in\Sigma}|\approxcount(\gamma,\setstr)-\counts(\gamma,\setstr)|\leq\frac{2\ell}{\epsilon_1}\ln(|\Sigma|/\beta_1)\leq \frac{2\ell}{\epsilon_1}\ln\left(\max\{\ell^2 n^2,|\Sigma|\}/\beta_1\right) 
$$ with probability at least $1-\beta_1$, while satisfying $\epsilon_1$-differential privacy. In the following, we let $\alpha=\frac{2\ell}{\epsilon_1}\ln(\max\{\ell^2 n^2,|\Sigma|\}/\beta_1)$. 
         We keep a pruned candidate set $\candidatespruned_{2^0}$ of strings of length 1 with a noisy count at least $\thresh=2\alpha$, i.e., $\candidatespruned_{2^0}$ consists of all $\gamma$ with $\approxcount(\gamma,\setstr)\geq \thresh$. If $|\candidatespruned_{2^0}|> n\ell$, we stop the algorithm and return a fail message.

\paragraph{Computing $\candidatespruned_{2^k}$ for $k=1,\dots,\lfloor \log \ell \rfloor$.} Given $\candidatespruned_{2^{k-1}}$ with $|\candidatespruned_{2^{k-1}}|\leq n\ell$ and $
        k\geq 1$, we compute $\candidatespruned_{2^k}$ as follows: first, we construct the set $\candidatespruned_{2^{k-1}}\circ \candidatespruned_{2^{k-1}}$ of the at most $(n\ell)^2$ strings of length $2^k$ that are a concatenation of two strings from $\candidatespruned_{2^{k-1}}$. Again by Corollary~\ref{cor:l1sens_same_length}, the $L_1$-sensitivity of $\counts$ for all strings in $\candidatespruned_{2^{k-1}}\circ \candidatespruned_{2^{k-1}}$ is bounded by $2\ell$.  We use the algorithm of  Corollary~\ref{cor:Laplace_mech} to estimate the counts of all strings in $\candidatespruned_{2^{k-1}}\circ \candidatespruned_{2^{k-1}}$ up to an additive error at most  $\frac{2\ell}{\epsilon_1}\ln(\ell^2 n^2/\beta_1)\leq \alpha$ with probability at least $1-\beta_1$ and $\epsilon_1$-differential privacy. The set $\candidatespruned_{2^k}$ is the set of all strings in $\candidatespruned_{2^{k-1}}\circ \candidatespruned_{2^{k-1}}$ with a noisy count at least $\thresh$. Again, if $|\candidatespruned_{2^k}|> n\ell$, we stop the algorithm and return a fail message.

\begin{example}\label{ex:candidatespruned}
    Consider the database $\setstr=\{$\texttt{aaaa, abe, absab, babe, bee, bees}$\}$ from Example~\ref{ex:doc_substr_count}. The exact versions, $\candidatespruned_{2^0}^{\times}$, $\candidatespruned_{2^1}^\times$, $\candidatespruned_{2^2}^\times$,  of the sets $\candidatespruned_{2^0}$, $\candidatespruned_{2^1}$, $\candidatespruned_{2^2}$ (i.e., the sets computed according to exact rather than noisy counts) with a frequency threshold $\tau=1$ are 
   \begin{align*}
       & \candidatespruned_{2^0}^\times=\{\texttt{a,b,e,s}\}\\
       & \candidatespruned_{2^1}^\times=\{\texttt{aa,ab,ba,be,bs,ee,es,sa}\}\\
       & \candidatespruned_{2^2}^\times=\{\texttt{aaaa,absa,babe,bees,bsab}\}
   \end{align*}
   The noisy, differentially private sets $\candidatespruned_{2^0}$, $\candidatespruned_{2^1}$, $\candidatespruned_{2^2}$ that are actually computed by the construction algorithm may contain (with sufficiently low probability) substrings that do not occur at all in $\setstr$; or conversely, some substrings occurring in $\setstr$ may be missing from the sets. A possible outcome of a noisy computation of these sets for $\setstr$ is
   \begin{align*}
       & \candidatespruned_{2^0}=\{\texttt{a,b,e,s}\}\\
       & \candidatespruned_{2^1}=\{\texttt{aa,ab,ba,be,bs,ee,sa}\}\\
       & \candidatespruned_{2^2}=\{\texttt{aaaa,absa,babe,bsab,aaab}\}
   \end{align*}  
   Note that the string \texttt{es} is missing from $\candidatespruned_{2^1}$ although it occurs once in $\setstr$, and consequently, \texttt{bees} is missing from $\candidatespruned_{2^2}$; conversely, \texttt{aaab} is in $\candidatespruned_{2^2}$ although it does not occur in $\setstr$.
\end{example} 
        
\paragraph{Constructing $\candidates$.} From $\candidatespruned_{2^0},\dots, \candidatespruned_{2^j}$, we now construct, for each $m$ which is not a power of two, a set of candidate patterns $\candidates_m$, without taking $\setstr$ further into account. 
        Specifically, for any fixed length $m$ with $2^k<m<2^{k+1}$, for $k=0,\dots, \lfloor \log \ell \rfloor$, we define $\candidates_m$ as the set of all strings $P$ of length $m$ such that $P[0,2^{k}-1$]$\in \candidatespruned_{2^k}$ and $P[m-2^{k},m-1]\in \candidatespruned_{2^k}$. For $m=2^k$ for some $k\in \{0,\dots, \lfloor \log \ell \rfloor\}$, we define $\candidates_m=\candidatespruned_{2^k}$. We finally define $\candidates=\bigcup_{m=1}^{\ell}\candidates_m$.
        
\begin{example}\label{ex:candidates}
    Consider the same database $\setstr$ and sets $\candidatespruned_{2^0}$, $\candidatespruned_{2^1}$, $\candidatespruned_{2^2}$ as in Example~\ref{ex:candidatespruned}. The set $\candidates_3$ is 
    \[\candidates_3=\{\texttt{aaa}, \texttt{aab}, \texttt{aba}, \texttt{abe}, \texttt{abs}, \texttt{baa}, \texttt{bab}, \texttt{bee}, \texttt{bsa}, \texttt{eee}, \texttt{saa}, \texttt{sab}\}\]
    $\candidates_3$ contains all possible strings of length $3$ whose length-$2$ prefix and suffix are in $\candidatespruned_{2^1}$; these strings can be obtained by taking all ordered pairs of strings $Q_1,Q_2\in\candidatespruned_{2^1}$ such that the length-$1$ suffix of $s_1$ is equal to the length-$1$ prefix of $s_2$: in this case, we say that $Q_1$ and $Q_2$ have a suffix/prefix overlap of length $1$. Note that many strings of $\candidates_3$ (e.g. \texttt{aba}) do not occur at all in $\setstr$ although, in this example, all strings of $\candidatespruned_{2^1}$ do. This is a normal behavior of our construction algorithm, and it is crucial to guarantee $\epsilon$-differential privacy.

    The set $\candidates_5$ is obtained from $\candidatespruned_{2^2}$ by taking all ordered pairs of strings with a suffix/prefix overlap of length $3$: we have
    \[\candidates_5=\{\texttt{aaaaa}, \texttt{aaaab},\texttt{absab}\}\]
\end{example}

We now prove that the algorithm above satisfies the following lemma.         
\begin{lemma}\label{lem:pruning}
    Let $\setstr=S_1,\dots, S_{n}$ be a database of documents. For any $\epsilon>0$ and $0<\beta<1$ there exists an $\epsilon$-differentially private algorithm, which computes a candidate set $\candidates\subseteq\Sigma^{[1,\ell]}$ enjoying the following two properties with probability at least $1-\beta$:
    \begin{itemize}
        \item For any $P\in \Sigma^{[1,\ell]}$ not in $\candidates$,  $\counts(P,\setstr)=O(\epsilon^{-1}\ell\log \ell \log (\max\{\ell^2 n^2,|\Sigma|\} /\beta))$,
        \item $|\candidates|\leq n^2\ell^3$.
    \end{itemize}
\end{lemma}
\begin{proof}
We use the algorithm described above. 

{\bf Privacy.} Since there are $\lfloor\log \ell\rfloor +1$ choices of  $k$, and for each we run an $\epsilon_1=\epsilon/(\lfloor\log \ell\rfloor +1)$-differentially private algorithm, their composition is $\epsilon$-differentially private by Lemma~\ref{lem:composition_theorem}. Since constructing $\candidates$ from $\candidatespruned_{2^0},\dots, \candidatespruned_{2^k}$ is post-processing, the entire algorithm is $\epsilon$-differentially private.
        
{\bf Accuracy.}  Since there are $\lfloor\log \ell\rfloor +1$ choices of $k$, by the choice of $\beta_1$, and by the union bound, all the error bounds hold together with probability at least $1-\beta$. Let $E$ be the event that all the error bounds hold simultaneously. In the following, we condition on $E$. Conditioned on $E$, for any $P\in \candidatespruned_{2^k}$, $k=0,\dots,\lfloor \log \ell \rfloor$, we have that $\counts(P,\setstr)\geq \thresh-\alpha\geq \alpha>1$, i.e., $P$ appears at least once as a substring in $\setstr$. Note that any string in $\setstr$ has at most $\ell$ substrings of length $2^k$. Since there are $n$ strings in $\setstr$, we have $|\candidatespruned_{2^k}|\leq n\ell$, conditioned on $E$. Thus, conditioning on $E$, the algorithm does not abort. Additionally, any $P$ of length $2^k$ which is not in $\candidatespruned_{2^k}$ satisfies $\counts(P,\setstr)<\thresh+\alpha=3\alpha$.
        Now, 
        any pattern $P$ of length $2^k<m<2^{k+1}$ for some $k\in\{0,\dots,\lfloor \log \ell \rfloor\}$ satisfies that $\counts(P,\setstr)\leq \counts(P[0, 2^k-1],\setstr)$ and $\counts(P,\setstr)\leq \counts(P[m-2^k, m-1],\setstr)$. Since $P\notin \candidates_m$ if and only if either $P[0, 2^k-1]\notin \candidatespruned_{2^k}$ or $P[m-2^k, m-1]\notin \candidatespruned_{2^k}$, we have that $P\notin \candidates_m$ implies $\counts(P,\setstr)<3\alpha$, conditioned on $E$. As $P[0,2^k-1]$ and $P[m-2^k, m-1]$ cover $P$ completely, we have $|\candidates_m|\leq |\candidatespruned_{2^k}|^2\leq (n\ell)^2$. Therefore, $|\candidates|=\sum_{m=1}^{\ell} |\candidates_m|\leq n^2 \ell^3$. 
    \end{proof}

The proof of Lemma~\ref{lem:C_construction_algo} describes an implementation of the algorithm above to compute $\candidates$ in polynomial time and space.

\begin{lemma}\label{lem:C_construction_algo}
    Given a database of $n$ documents $\setstr=S_1,\dots, S_{n}$ over $\Sigma^{[1,\ell]}$, a set $\candidates$ satisfying the properties of Lemma~\ref{lem:pruning} can be computed in time $O\left(n^2\ell^3\log\log(n\ell) +n^2\ell^4\right)$ 
    and space $O(n^2\ell^4)$ .
\end{lemma}
\begin{proof}
We construct the suffix tree $\ST$ of the string $S = S_1 \$_1S_2 \$_2\ldots S_n \$_n$, where $\$_1,\ldots,\$_n\notin \Sigma$, and store, within each branching node $v$, the ID $\leaf(v)$ of the leftmost descending leaf.
We also construct a substring concatenation data structure over $S$. A substring concatenation query consists of four integers $i_1, i_2, j_1, j_2$. If $S[i_1,j_1]\circ S[i_2,j_2]$ is a substring of $S$ it returns a pair of indices $i,j$ such that $S[i,j] = S[i_1,j_1]\circ S[i_2,j_2]$, together with a pointer to the shallowest (i.e., closest to the root) branching node $v$ of $\ST$ such that $S[i,j]$ is a prefix of $\str{v}$; it returns a $\nil$ pointer otherwise. Such a data structure can be constructed in $O(n\ell)$ time and space, and answers queries in $O(\log\log (n\ell))$ time~\cite{DBLP:journals/algorithmica/BilleCCGSVV18, BGVV2014}. 

\textbf{Main algorithm.} The algorithm proceeds in phases for increasing values of $k=0,1,\ldots,\lfloor \log \ell \rfloor$. In Phase $k=0$, it computes a noisy counter $c_\gamma^*$ for each letter $\gamma\in\Sigma$, and stores in $\candidatespruned_{2^0}$ all and only the letters such that $c_\gamma^*\geq\tau$. To do so, it first computes the true frequency of each letter by traversing the suffix tree $\ST$: the frequency of $\gamma$ is given by $f(v_\gamma)$, where $v_\gamma$ is the shallowest branching node of $\ST$ such that $\str{v_\gamma}$ starts with $\gamma$ (the frequency is $0$ if no such node exists, i.e., $\gamma$ does not occur in the database $\setstr$). Noisy counts $c_\gamma^*$ are then obtained as in Lemma~\ref{lem:pruning}. $\candidatespruned_{2^0}$ is represented by a list of pairs $\langle\gamma,p\rangle$, where $p$ is a pointer to $v_\gamma$ ($p=\nil$ if $v_\gamma$ does not exist).

In any phase $k>0$, the algorithm 
performs two steps: (1) it computes a noisy counter $c_Q^*$ for every string $Q\in\candidatespruned_{2^{k-1}}\circ \candidatespruned_{2^{k-1}}$ and stores a representation of $\candidatespruned_{2^{k}}$ consisting of all and only the strings $Q$ such that $c_Q^*\geq\tau$ (it aborts the procedure and returns FAIL if $|\candidatespruned_{2^{k}}|>n\ell$); (2) it constructs sets $\candidates_m$ for every $2^k<m<2^{k+1}$.  
We now describe each of these steps in more detail.

\textbf{Step (1)}~  $\candidatespruned_{2^{k}}$ is represented as a list of pairs $\langle Q,p\rangle$, where $Q$ is a string of length $2^k$ and $p$ is a pointer to the shallowest node $v$ of $\ST$ such that $Q$ is a prefix of $\str{v}$ (or $p=\nil$ if no such $v$ exists). Consider all ordered pairs of items from the $\candidatespruned_{2^{k-1}}$ list. 
For each such pair $\left(\langle Q_1,p_1\rangle,\langle Q_2,p_2\rangle \right)$, if $p_1,p_2\neq\nil$, let $v_1,v_2$ be the nodes of $\ST$ $p_1$ and $p_2$ point to, respectively. We ask a concatenation query with $i_1=\leaf(v_1)$, $i_2=\leaf(v_2)$, $j_1=i_1+2^{k-1}-1$, $j_2=i_2+2^{k-1}-1$; if the result is a pointer $p$ to a node $v$ of $\ST$, we compute $c_{Q_1Q_2}^*$ by adding noise to $f(v)$ as in Lemma~\ref{lem:pruning}, append the pair $\langle Q_1Q_2,p\rangle$ to the list of $\candidatespruned_{2^{k}}$ if $c_{Q_1Q_2}^*\geq\tau$, or discard it and move on to the next pair of elements from $\candidatespruned_{2^{k-1}}$ otherwise. In all the other cases, that is, if either $p_1=\nil$ or $p_2=\nil$ or the result of the concatenation query is $p=\nil$, the true frequency of $Q_1\cdot Q_2$ in $\setstr$ is $0$, thus we compute $c_{Q_1Q_2}^*$ by adding noise to $0$ and append $\langle Q_1Q_2,p\rangle$ to the list of $\candidatespruned_{2^{k}}$ only if $c_{Q_1Q_2}^*\geq\tau$. We stop the procedure and return FAIL whenever the number of triplets in $\candidatespruned_{2^{k}}$ exceeds $n\ell$.

\textbf{Step (2)}~ For any $2^k<m<2^{k+1}$, set $\candidates_m$ consists of all strings $Q$ of length $m$ whose prefix $Q_1=Q[0,2^k-1]$ and suffix $Q_2=Q[m-2^k,m-1]$ are both in $\candidatespruned_{2^k}$. Note that, since $m<2^{k+1}$, this implies that $Q_1$ and $Q_2$ have a suffix/prefix overlap of length $2^{k+1}-m$, or in other words, $Q_1[m-2^k,2^k-1]$ and $Q_2[0,2^k-1]$ have a common prefix of length $2^{k+1}-m$. The construction of $\candidates_m$ thus reduces to finding all pairs of strings from $\candidatespruned_{2^k}$ with a suffix/prefix overlap of length $2^{k+1}-m$. To do this efficiently for all $m$, we preprocess $\candidatespruned_{2^k}$ to build in $O(\textrm{sort}(2^k n\ell, |\Sigma|))$ time a data structure that occupies $O(2^k n\ell)$ space and answers \emph{Longest Common Extension} (LCE) queries in $O(1)$ time~\cite{FCFM2000,harel1984fast, DBLP:conf/latin/BenderF00} (since $2^k n\ell$ upper bounds the total length of the strings in $\candidatespruned_{2^k}$). An LCE query consists of a pair of strings $Q_1,Q_2\in \candidatespruned_{2^k}$ and two positions $q_1,q_2< 2^k$; the answer $\textrm{LCE}_{Q_1,Q_2}(q_1,q_2)$ is the length of the longest common prefix of the suffixes $Q_1[q_1,2^k-1]$ and $Q_2[q_2,2^k-1]$.

Once the LCE data structure is constructed, consider all ordered pairs of items from the list of $\candidatespruned_{2^{k}}$. 
For each pair $\left(\langle Q_1,p_1\rangle,\langle Q_2,p_2\rangle\right)$ and each $2^k<m<2^{k+1}$, compute $\textrm{LCE}_{Q_1,Q_2}(m-2^k,0)$: if the result is $2^{k+1}-m$, then $Q_1$ and $Q_2$ have a suffix-prefix overlap of length $2^{k+1}-m$, thus the concatenation of $Q_1$ and the suffix of $Q_2$ immediately following the overlap, that is $Q=Q_1[0,2^k-1]\cdot Q_2[2^{k+1}-m,2^k-1]$, belongs to $\candidates_m$.
Once again, we represent $Q_m$ with a list of pairs $\langle Q,p\rangle$, where $p$ is a pointer to the shallowest node $v$ of $\ST$ such that $Q$ is a prefix of $\str{v}$, with $p=\nil$ if $Q$ does not appear in $\setstr$ (thus no such node $v$ exists). The pointer $p$ for a string $Q=Q_1[0,2^k-1]\cdot Q_2[2^{k+1}-m,2^k-1]$ is obtained as the result of the concatenation query with $i_1=\leaf(v_1)$, $i_2=\leaf(v_2)+2^{k+1}-m$, $j_1=i_1+2^k-1$, $j_2=i_2+2^k-1$; $p=\nil$ if either $p_1=\nil$ or $p_2=\nil$.

\textbf{Time and space analysis.} Consider Step 1 and focus first on the space occupancy. By construction, at any phase $k=0,\ldots,\lfloor \log \ell \rfloor$, $\candidatespruned_{2^k}$ is represented by a list of $O(n\ell)$ pairs, each consisting of a string of length $2^k$ and a pointer. Storing the list in Phase $k$ thus requires space $O(2^kn\ell)$, implying a total space $O(n\ell^2)$ to represent all the lists. As for the running time, in each phase $k=1,\ldots,\lfloor \log \ell \rfloor$ we consider $O(n^2\ell^2)$ pairs of elements from $\candidatespruned_{2^{k-1}}$; for each pair, we ask one concatenation query, requiring $O\left(\log\log (n\ell)\right)$ time (as to ask the query we do not need to read the string of the pair), compute a noisy counter in $O(1)$ time, and write a new pair consisting of a string of length $2^{k+1}$ and a pointer if the counter exceeds $\tau$. Since we abort the procedure if the number of written pairs exceeds $n\ell$, the total writing time in Phase $k$ is $O(2^{k+1}n\ell)$, thus $O\left(\sum_{k=0}^{\lfloor \log\ell\rfloor} 2^{k+1}n\ell\right)=O(n\ell^2)$ over all phases. Therefore, the total time for Step 1 is dominated by the concatenation queries, and summing over the $\lfloor \log\ell\rfloor+1$ phases it is $O(n^2\ell^2\log\log(n\ell)\log\ell)$.

Let us now focus on Step 2, and consider first the space occupancy. For each phase $k=0,\ldots,\lfloor\log\ell\rfloor$, we construct an LCE data structure that occupies $O(2^k n\ell)$ space; however, at the end of Phase $k$, this data structure can be discarded, thus the space it occupies at any point of the algorithm is $O\left(n\ell^2\right)$. For each $2^k<m<2^{k+1}$, we represent $\candidates_m$ with a list of up to $n^2\ell^2$ pairs consisting of one string of size $m$ and one pointer, thus using space $O\left(mn^2\ell^2\right)$. The total space to store all lists is thus $O\left(\sum_{m=1}^\ell mn^2\ell^2\right)=O(n^2\ell^4)$, which dominates the space occupancy of the algorithm.
The running time of Step 2 is as follows. For each $k$, we construct the LCE data structure for $\candidatespruned_{2^k}$ in time $O\left(\textrm{sort}(2^kn\ell,|\Sigma|)\right)$, thus the total time for all phases is $O\left(\sum_{k=0}^{\lfloor\log\ell\rfloor}\left(\textrm{sort}(2^kn\ell,|\Sigma|)\right)\right)=O\left(\textrm{sort}(n\ell^2,|\Sigma|)\right)$. To compute the list representation of $\candidates_m$, we ask $O\left(n^2\ell^2\right)$ LCE and concatenation queries in total $O\left(n^2\ell^2\log\log(n\ell)\right)$ time; 
and we use $O\left(mn^2\ell^2\right)$ time to write all the list pairs. Summing over all $m=1,\ldots,\ell$, we obtain a total time $O\left(n^2\ell^4\right)$ to write all the lists and $O\left(n^2\ell^3\log\log(n\ell)\right)$ total time for the LCE and concatenation queries. The total running time for Step 2 is thus $O\left(\textrm{sort}(n\ell^2,|\Sigma|) + n^2\ell^3\log\log(n\ell) +n^2\ell^4\right) = O\left(n^2\ell^3\log\log(n\ell) +n^2\ell^4\right)$. 
\end{proof}


\subsection{Steps 2 to 4: Heavy Path Decomposition and Difference Sequences}\label{subsec:heavypath}
The next steps are to arrange $\candidates$ in a trie $\candidatetrie$ and compute its heavy path decomposition.
Note that the number of nodes $|\candidatetrie|$ is bounded by the total length of strings in $\candidates$, which is in turn bounded by $n^2\ell^4$. 
This part of the algorithm exploits the way the counts associated with nodes of $\candidatetrie$ can vary in neighboring databases. 
Before describing step 2 to 4 in detail we will first prove some properties on counts of nodes on paths in $T_C$.

\subsubsection{Properties of counts on paths.}
A first key observation is as follows: Consider the counts of the string $\str{v}$ associated with any node $v$ of $\candidatetrie$ in two neighboring databases $\setstr$ and $\setstr' = \setstr\setminus\{S\}\cup\{S'\}$.
The difference of the counts of $\str{v}$ in $\setstr$ and $\setstr'$
depends only on how many times $\str{v}$ occurs in the strings $S$ and $S'$ by which $\setstr$ and $\setstr'$ differ:
\begin{observation}\label{obs:sensitvity_node}
     For any node $v$ in $\candidatetrie$ and any two neighboring databases $\setstr$ and $\setstr' = \setstr\setminus\{S\}\cup\{S'\}$, $|\counts(\str{v}, \setstr)-\counts(\str{v}, \setstr')|= |\counts(\str{v},S) - \counts(\str{v},S')|$.
\end{observation}
 

Let us now focus on how the counts computed in neighboring databases can vary along paths of $\candidatetrie$. For any path $p=v_0,v_1,\dots, v_{|p|-1}$ in the trie $\candidatetrie$, we define the \emph{difference sequence of $\counts$ on $p$} as the $(|p|-1)$-dimensional vector $\diff_p(\setstr)[i]=\counts(\str{v_i},\setstr)-\counts(\str{v_{i-1}},\setstr)$ for $i=1\dots |p|-1$: see Figure~\ref{fig:binary_tree_mechanism}(a) for an example. Note that the counts associated to the nodes of a path are non-increasing when descending the path: this is because, if $v_j$ is a descendant of $v_i$, then $\str{v_i}$ is a prefix of $\str{v_j}$, thus $\counts(\str{v_i},S)\geq \counts(\str{v_j},S)$. This allows us to prove
the following lemma.

\begin{lemma}\label{lem:sensitvity_path} Let $\setstr$ and $\setstr' = \setstr\setminus\{S\}\cup\{S'\}$ be two neighboring databases. Then,
    for any path $p$ with root $r$ in $\candidatetrie$, $||\diff_p(\setstr)-\diff_p(\setstr')||_1\leq \counts(\str{r},S)+ \counts(\str{r},S')$.
\end{lemma}
\begin{proof} We have
\begin{align*}
&||\diff_p(\setstr)-\diff_p(\setstr')||_1 \\&=\sum_{i=1}^{|p|-1} |\counts(\str{v_i},S)-\counts(\str{v_{i-1}},S)
 -  (\counts(\str{v_i},S')- \counts(\str{v_{i-1}},S')) |\\
& \leq \sum_{i=1}^{|p|-1} |\counts(\str{v_i},S)-\counts(\str{v_{i-1}},S)| 
 + \sum_{i=1}^{|p|-1}|\counts(\str{v_i},S')-\counts(\str{v_{i-1}},S')|.
\end{align*}
Note that if $v'$ is a descendant of $v$, then $\str{v}$ is a prefix of $\str{v'}$ and therefore $\counts(\str{v},S)\geq \counts(\str{v'},S)$. 
Since $\counts(\str{v_i},S)$ is monotonically non-increasing in $i$ for $p=v_0, v_1,\dots, v_{|p|-1}$ with $v_0=r$, 
from $0\leq \counts(\str{v_i},S)\leq \counts(\str{r},S)$, for any descendant $v_i$ of $r$, it follows that $\sum_{i=1}^{|p|-1} |\counts(\str{v_i},S)-\counts(\str{v_{i-1}},S)|\leq \counts(\str{r},S)$. The same argument applies for $S'$ and thus the lemma follows.
\end{proof}

We now pair these properties with the well-known \emph{heavy path decomposition} of a tree $T$, defined as follows.
Every edge is either \emph{light} or \emph{heavy}. There is exactly one heavy edge outgoing from every node except the leaves, defined as the edge to the child whose subtree contains the most nodes; ties are broken arbitrarily. 
The longest \emph{heavy path} is obtained by following the heavy edges from the root of $T$ to a leaf: note that all edges hanging off this path are light. The other heavy paths are obtained by recursively decomposing all subtrees hanging off a heavy path. We call the topmost node of a heavy path (i.e., the only node of the path which is not reached by a heavy edge) its \emph{root}. The heavy path decomposition of a tree with $N$ nodes can be constructed in $O(N)$ time and has the following important property~\cite{SLEATOR1983362}:
    \begin{lemma}
    \label{lem:heavy_path} Given a tree $T$ with $N$ nodes and a heavy path decomposition of  $T$, any root-to-leaf path in $T$ contains at most $\lfloor\log N\rfloor$ light edges. 
    \end{lemma}

\begin{example}\label{ex:heavy_paths}
    Consider the candidate set $\candidates=\candidatespruned_{2^0}\cup\candidatespruned_{2^1}\cup\candidatespruned_{2^2}\cup\candidates_3\cup \candidates_5$ computed in Examples~\ref{ex:candidatespruned} and~\ref{ex:candidates}. The trie $\candidatetrie$ of $\candidates$ with its heavy path decomposition is depicted in Figure~\ref{fig:heavy_paths}.
\end{example}
 \begin{figure}[t]
     \centering
     \includegraphics[width=.58\columnwidth]{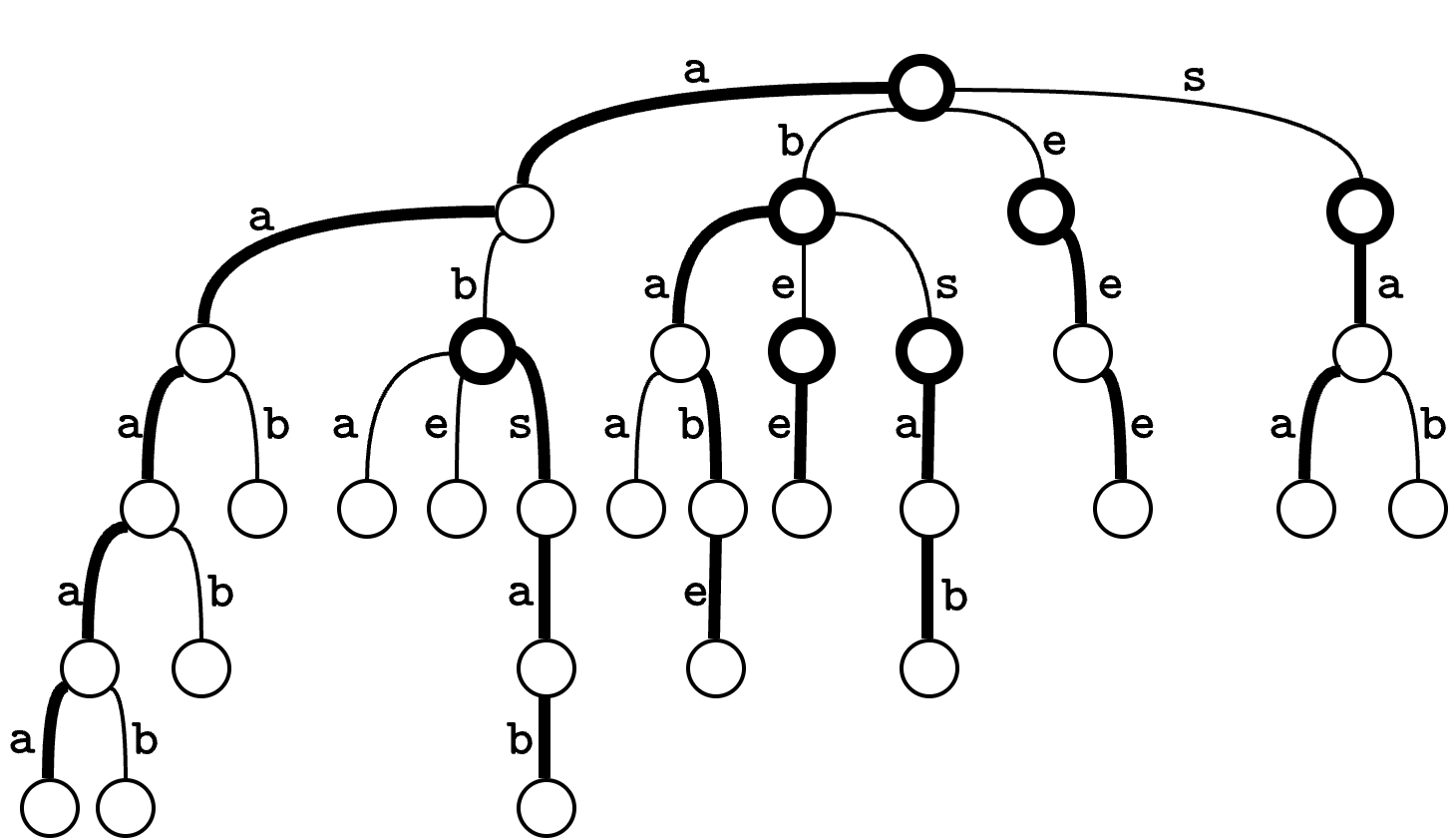}
     \caption{The trie of the candidate set from Example~\ref{ex:candidates}. Bold edges are heavy; bold nodes are the roots of the heavy paths.}
     \label{fig:heavy_paths}
 \end{figure}

We now exploit Lemma~\ref{lem:heavy_path}, together with the fact that any string $S$ of length at most $\ell$ (not necessarily included in the database from which $\candidatetrie$ has been constructed) can only influence counts over the at most $\ell$ paths of $\candidatetrie$ corresponding to its suffixes, to arrive at Lemma~\ref{lem:all_roots}.

\begin{lemma}\label{lem:all_roots}
    Let $r_0,r_1,\dots, r_k$ be the roots of the paths of the heavy path decomposition of $\candidatetrie$, where $r_0$ is the root of $\candidatetrie$. Then for any string $S$ of length at most $\ell$, we have $\sum_{i=0}^k \counts(\str{r_i},S)\leq \ell(\lfloor\log |\candidatetrie|\rfloor +1)=O(\ell \log (n\ell))$.
\end{lemma}
\begin{proof}
    For any suffix $S_i$ of $S$, let $v_i$ be the node in $\candidatetrie$ such that $\str{v_i}$ is the longest prefix of $S_i$ which is in $\candidates$. 
    We say that the path from the root of $\candidatetrie$ to $v_i$ \emph{corresponds} to $S_i$.
    For any heavy path root $r$, note that $\counts(\str{r}, S)$ is exactly the number of suffixes of $S$ that begin with $\str{r}$. Further, if a suffix $S_i$ has $\str{r}$ as a prefix, then the path corresponding to $S_i$ goes through $r$. 
    By the property of the heavy path decomposition (Lemma~\ref{lem:heavy_path}), the path corresponding to $S_i$ contains at most $\lfloor\log |\candidatetrie|\rfloor$ 
    light edges,  
    and as such, at most $\lfloor\log |\candidatetrie|\rfloor+1$ 
    heavy path roots. Thus, any suffix $S_i$ of $S$ contributes at most $\lfloor\log |\candidatetrie|\rfloor+1$ 
    to the total sum. Since we have at most $\ell$ suffixes, this gives us $\sum_{i=0}^k\counts(\str{r_i},S)\leq \ell(\lfloor\log |\candidatetrie|\rfloor +1)=O(\ell \log (n\ell))$.
%
\end{proof}


\subsubsection{Steps 2 and 3}
We are now ready to describe step 2 and 3 of the algorithm.

\paragraph{Step 2.} We build the trie $\candidatetrie$ of the candidate set $\candidates$ constructed in Section~\ref{sec:construct_P_epsilon_DP}, and construct a heavy path decomposition of $\candidatetrie$.  

  \paragraph{Step 3.} For every root $r$ of a heavy path, we compute a differentially private estimate of $\counts(\str{r},\setstr)$ as follows. 

Observation~\ref{obs:sensitvity_node} and Lemma~\ref{lem:all_roots} give that the $L_1$-sensitivity of $\left(\counts(\str{r_i},\setstr)\right)_{i=0}^k$ is bounded by $\ell(\lfloor\log |\candidatetrie|\rfloor +1)=O(\ell \log (n\ell))$. Additionally note that, since $\candidatetrie$ has at most $n^2\ell^3$ leaves, there are at most $n^2\ell^3$ heavy paths. Together with Corollary~\ref{cor:Laplace_mech}, this gives an $\epsilon$-differentially private algorithm to compute noisy counts for the heavy paths roots: 
\begin{corollary}\label{cor:heavypathroots}
Let $r_0,\dots, r_k$ be the roots of the heavy paths of $\candidatetrie$. 
    For any $\epsilon>0$ and $0<\beta<1$, there exists an $\epsilon$-differentially private algorithm, which estimates $\counts(\str{r_i}, \setstr)$ for the heavy path roots $r_0,\dots, r_k$ up to an additive error at most $O(\epsilon^{-1}\ell \log (n\ell)\ln (k/\beta))=O(\epsilon^{-1}\ell\log^2 (n\ell/\beta))$ with probability at least $1-\beta$.
\end{corollary}
Additionally, Lemma~\ref{lem:sensitvity_path} and Lemma~\ref{lem:all_roots} give that the \emph{sum of} $L_1$-sensitivities of the difference sequences \emph{over all heavy paths} is bounded by $2\ell(\lfloor\log |\candidatetrie|\rfloor +1)$. As the next step, we show how to compute all \emph{prefix sums} of the difference sequences of all heavy paths with an error $O(\ell\lfloor\log |\candidatetrie|\rfloor )$ up to polylogarithmic terms. 

\subsubsection{Computing noisy prefix sums using the binary tree mechanism.}
To compute the noisy prefix sums, we show that we can estimate the prefix sums of $k$ sequences while preserving differential privacy up to an additive error which is roughly the sum of their $L_1$-sensitivities. The algorithm builds a copy of the binary tree mechanism \cite{Dwork2010} for each of the $k$ sequences. 
The binary tree mechanism works as follows: To compute a differentially private estimate of all prefix sums of a finite sequence, we first divide the sum into partial sums of power-of-two many elements, and add noise to each partial sum. That is, to compute all prefix sums of a sequence $a_1,\dots, a_Q$, we first compute a noisy value of each $a_i$, then of $a_1+a_2$, $a_3+a_4$, $\dots$, then of $a_1+a_2+a_3+a_4$, $a_5+a_6+a_7+a_8$, and so on. The main idea is that any value $a_i$ can influence at most $\log Q$ of these partial sums, and we can compute an estimate of any prefix sum by summing at most $\log Q$ of the noisy partial sums. In the following lemma, we analyze this idea more generally for $k$ sequences $a^{(1)},\dots, a^{(k)}$:

\begin{lemma}\label{lem:binary_trees}
Let $\chi^{*}$ be any universe of possible data sets and $T\in\mathbb{N}$. Let $a^{(1)}, \dots, a^{(k)}$ be $k$ functions, where the output of every function is a sequence of length $T$, i.e. $a^{(i)}:\chi^{*}\rightarrow \mathbb{N}^{T}$ for all $i\in[1,k]$. 
Let $L$ be the sum of $L_1$-sensitivities of $a^{(1)}, \dots, a^{(k)}$, that is, let $L=\max_{x\sim x'}\sum_{i=1}^k \sum_{j=1}^{T}|a^{(i)}(x)[j]-a^{(i)}(x')[j]|$. For any $\epsilon>0$ and $0<\beta<1$ there exists an $\epsilon$-differentially private algorithm computing for every $i\in[1,k]$ all prefix sums of $a^{(i)}(x)$ with additive error at most $O(\epsilon^{-1}L\log T \log(Tk/\beta))$ with probability at least $1-\beta$.
\end{lemma}
For the proof, we need the following lemma on the sum of Laplace random variables.

\begin{lemma}[Sum of Laplace Variables]\label{lem:sum_of_eq_lap}
Let $Y_1,\dots,Y_k$ be independent variables with distribution $\Lap(b)$ and let $Y=\sum_{i=1}^k Y_i$. Let $0<\beta<1$. Then
\begin{align*}  \Pr\left[|Y|>2b\sqrt{2\ln(2/\beta)}\max\left\{\sqrt{k},\sqrt{\ln(2/\beta
    )}\right\} \right] \leq \beta.
\end{align*}
\end{lemma}
\begin{proof}
Apply Corollary 2.9 in \cite{DBLP:journals/tissec/ChanSS11} to $b_1=\dots=b_k=b$.
\end{proof}

\begin{proof}[Proof of Lemma~\ref{lem:binary_trees}]
The algorithm builds a copy of the binary tree mechanism \cite{Dwork2010} for each of the $k$ sequences. 
The dyadic decomposition of an interval $[1,T]$ is given by the set $\mathcal{I}=\{[j\cdot 2^i+1, (j+1)2^{i}],0\leq j\leq \lceil T/2^i\rceil -1, 0\leq i \leq \lfloor \log T \rfloor\}$. It has the property that every interval $[1, m]$ for $m\leq T$ is the union of no more than $\lfloor\log T\rfloor+1$ disjoint intervals from $\mathcal{I}$. Denote these intervals by $\mathcal{I}_{[1,m]}$.  The algorithm is now as follows: For every sequence $i \in [1,k]$ and interval $[b,e]\in\mathcal{I}$, we independently draw a random variable $Y_{[b,e]}^{(i)}\approx\Lap(\epsilon^{-1}L(\lfloor \log T \rfloor +1))$. Given the sequences $a^{(1)}, \dots, a^{(k)}$ and the data set $x$ we compute for every sequence $a^{(i)}(x)$, and every interval $[b,e]\in\mathcal{I}$, the partial sum $s^{(i)}_{[b,e]}(x)=a^{(i)}(x)[b]+\dots+a^{(i)}(x)[e]$. For each such partial sum, we compute an approximate sum $\tilde{s}^{(i)}_{[b,e]}=s^{(i)}_{[b,e]}(x)+Y_{[b,e]}^{(i)}$. We now compute the approximate $m$th prefix sum of the $i$th sequence by $\sum_{[b,e]\in \mathcal{I}_{[1,m]}} \tilde{s}^{(i)}_{[b,e]}$, for all $i=1,\dots, k$ and all $m=1,\dots, T$.

{\bf Privacy.} We show that the $L_1$-sensitivity of the sequence of $s^{(i)}_{[b,e]}(x)$  with $i=1,\dots k$ and $[b,e]\in \mathcal{I}$ is at most $L(\lfloor \log T \rfloor +1)$: Recall that $L\geq\max_{x\sim x'}\sum_{i=1}^k \sum_{j=1}^{T}|a^{(i)}(x)[j]-a^{(i)}(x')[j]|$. The $L_1$-sensitivity for all $s^{(i)}_{[b,e]}(x)$ is bounded by 
\begin{align*}
    \max_{x\sim x'}\sum_{i=1}^k\sum_{[b,e]\in \mathcal{I}}|s^{(i)}_{[b,e]}(x)-s^{(i)}_{[b,e]}(x')|&\leq  \max_{x\sim x'}\sum_{i=1}^k\sum_{[b,e]\in \mathcal{I}}\sum_{j\in [b,e]}|a^{(i)}(x)[j]-a^{(i)}(x')[j]|\\&=\max_{x\sim x'}\sum_{i=1}^k\sum_{j=1}^{T}\sum_{[b,e]\in \mathcal{I}:j\in [b,e]}|a^{(i)}(x)[j]-a^{(i)}(x')[j]|\\&\leq L(\lfloor \log T \rfloor +1),
\end{align*}
where the last inequality follows from the fact that any $j\in[1,T]$ appears in no more than $\lfloor \log T \rfloor +1$ many sets in $\mathcal{I}$. Thus, the algorithm above is $\epsilon$-differentially private by Lemma~\ref{lem:Laplacemech}.

{\bf Accuracy.} To prove the error bound, we use Lemma~\ref{lem:sum_of_eq_lap} for sums of Laplace variables: For any prefix sum, we sum at most $\lfloor \log T\rfloor +1$ Laplace variables of scale $\epsilon^{-1}L(\lfloor\log T\rfloor +1)$. Let $Y^{(i)}_{[1,m]}$ be the sum of Laplace variables used to compute $\sum_{[b,e]\in \mathcal{I}_{[1,m]}} \tilde{s}^{(i)}_{[b,e]}$, i.e., $Y^{(i)}_{[1,m]}=\sum_{[b,e]\in \mathcal{I}_{[1,m]}}Y^{(i)}_{[b,e]}$. By Lemma~\ref{lem:sum_of_eq_lap} we get
\begin{align*}
    \Pr\left[|Y^{(i)}_{[1,m]}|>2\epsilon^{-1}L(\lfloor\log T\rfloor +1)\sqrt{2\ln(2kT/\beta)}\max\left\{\sqrt{(\lfloor\log T\rfloor +1)}, \sqrt{\ln(2kT/\beta)}\right\}\right]\leq \beta/(kT).
\end{align*}
By a union bound over at most $kT$ prefix sum computations, we get that with probability at least $1-\beta$, the error for any prefix sum is bounded by $O(\epsilon^{-1}L\log T \log(kT/\beta))$.
\end{proof}

\subsubsection{Step 4}
For a heavy path $p$ let $\mathrm{sums}_p$ be all the prefix sums of its difference sequence, i.e., $\mathrm{sums}_p[i]= \sum_{j=1}^i \diff_p(\setstr)[j]$. 
 In step 4 we compute for 
  every heavy path $p=v_0,v_1,\dots, v_{|p|-1}$ with root $r=v_0$, 
  a differentially private estimate of all \emph{prefix sums}, $\mathrm{sums}^*_p$, of its difference sequence. That is, $\mathrm{sums}^*_p[i]$ is a noisy estimate of $\sum_{j=1}^i \diff_p(\setstr)[j]$.
  We compute $\mathrm{sums}^*_p$ using the algorithm from Lemma~\ref{lem:binary_trees} with $L=\ell(\lfloor\log |\candidatetrie|\rfloor +1)=O(\ell \log (n\ell))$, $k\leq n^2\ell^3$, and $T=\ell$. This gives us the following corollary.

\begin{corollary}\label{cor:heavypaths}
    Let $p_0,\dots, p_k$ be all heavy paths of $\candidatetrie$. For any $\epsilon>0$ and $0<\beta<1$, there exists an $\epsilon$-differentially private algorithm, which estimates $\sum_{j=1}^i \diff_{p_m}(\setstr)[j]$ for all $i=1,\dots, |p_m|$ and all $m=0,\dots, k$ up to an additive error at most $$O\left(\epsilon^{-1} \ell\log (n\ell)\log \ell \log(\ell k/\beta)\right)=O\left(\epsilon^{-1}\ell\log \ell \log^2 (n\ell/\beta)\right)$$ with probability at least $1-\beta$.
\end{corollary}

A schema of the application of Corollary~\ref{cor:heavypaths} to one of the heavy paths of the trie of Figure~\ref{fig:heavy_paths} is in Figure~\ref{fig:binary_tree_mechanism}.

\begin{figure}
    \centering
    \includegraphics[width=.85\linewidth]{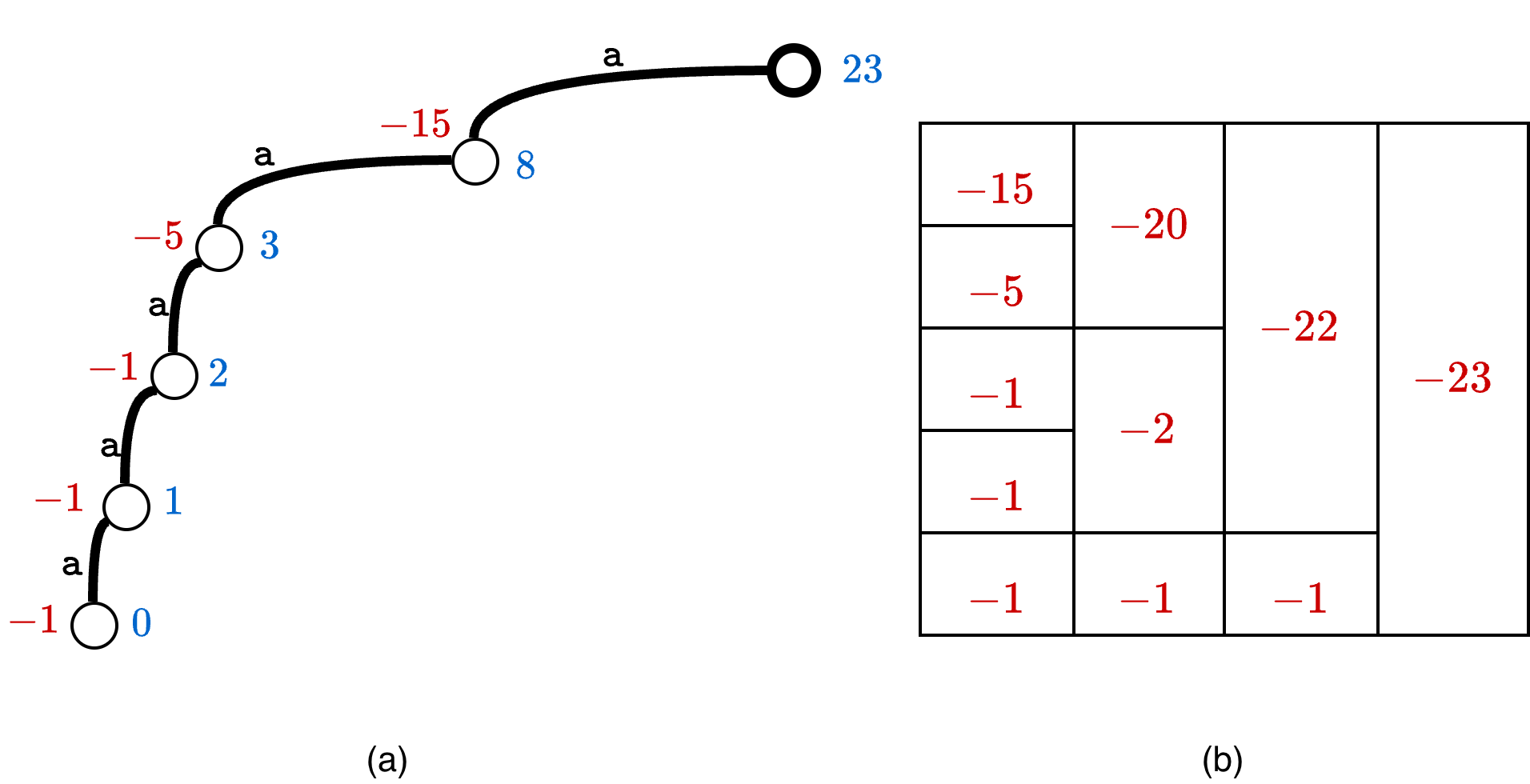}
    \caption{(a) The heavy path rooted at the root of the trie of Figure~\ref{fig:heavy_paths};  the real counts are in blue, the difference sequence is in red. The root represents the empty string, assumed to occur at every position of the database $\setstr$ from Example~\ref{ex:candidates}. (b) The partial sums of the difference sequence (without noise). The algorithm computes a noisy version of the values in the table shown.}
    \label{fig:binary_tree_mechanism}
\end{figure}

\subsubsection{Efficient implementation.}
   Lemma~\ref{lem:heavy_decomp_complexity} below provides a procedure to compute step 2-4 efficiently.

    \begin{lemma}\label{lem:heavy_decomp_complexity}
        Given the representation of $\candidates$ output by the algorithm from Section~\ref{sec:construct_P_epsilon_DP} for a database $\setstr$ of $n$ documents from $\Sigma^{[1,\ell]}$, a trie $T_\candidates$ and its heavy path decomposition, noisy counts of the heavy path roots and noisy prefix sums on the heavy paths, can be computed in $O(n^2\ell^4)$ time and space.
    \end{lemma}

\begin{proof}
    Let $\setstr=S_1,\dots, S_n$. Given the representation of $\candidates$ output by the algorithm from Section~\ref{sec:construct_P_epsilon_DP} with the suffix tree of $S_1\$_1\dots S_n\$_n$, we can build the trie augmented with the true counts for each node in $O(n^2\ell^4)$ time and space: we just insert each string letter by letter, and at the same time traverse the suffix tree to compute the counts for each prefix. Given the trie, we can compute the heavy path decomposition in linear time of the size of the trie, i.e. in $O(n^2\ell^4)$ time. Given the counts for each node, we can compute the noisy counts for all roots of the heavy paths in $O(n^2\ell^3)$ time. Further, for every heavy path of length $h$, we can compute the noisy prefix sums of the difference sequences in $O(h)$ time and space: note that we divide the length of the path into dyadic intervals, and need to compute a noisy count for each dyadic interval. Constructing these partial sums bottom-up can be done in $O(h)$ time. Since all heavy paths lengths summed up are bounded by the size of the tree, the total time is $O(n^2\ell^4)$.
    \end{proof}

\subsection{Steps 5 and 6: Putting It Together}

Finally, we compute the noisy counts for each node and prune the trie as follows. For a heavy path $p= v_0, v_1,\dots, v_{|p|-1}$ with root $r=v_0$, let $\approxcount(\str{r},\setstr)$ be the noisy count computed for the root, and let $\mathrm{sums}^*_p[i]$ be the noisy estimate of $\sum_{j=1}^i \diff_p(\setstr)[j]$. Then for any node $v_i$ on the path with $i>0$, we compute $\approxcount(\str{v_i},\setstr)=\approxcount(\str{r},\setstr)+\mathrm{sums}^*_p[i]$.

 {\bf Parameters.} To obtain an $\epsilon$-differential private algorithm, we set $\epsilon'=\epsilon/3$ and failure probability $\beta'=\beta/3$ in Steps 1, 3, and 4 of the algorithm.
 
Let $\alpha$ be the sum of the error bounds for Corollary~\ref{cor:heavypathroots} and Corollary~\ref{cor:heavypaths} which each hold with probability $1-\beta'$. We have $\alpha=O(\epsilon^{-1}\ell\log \ell \log^2(n\ell/\beta))$. As a final step, we prune $\candidatetrie$ by traversing the trie from the root, and for any node $v$ with $\approxcount(\str{v_i},\setstr)<2\alpha$, we delete $v$ and its subtree. We return the resulting pruned trie $T^*$ together with approximate counts $\approxcount(\str{v},\setstr)$ for every node $v 
\in T^*$.

By the composition theorem (Lemma~\ref{lem:composition_theorem}), this algorithm is $\epsilon$-differentially private. By a union bound, all its error guarantees hold together with probability $1-3\beta'=1-\beta$. Thus, we get with probability $1-\beta$ a trie $T^*$, with the following properties:
\begin{itemize}
    \item For each node $v\in T^*$, by Corollaries~\ref{cor:heavypathroots} and \ref{cor:heavypaths},
    $$|\approxcount(\str{v},\setstr)-\counts(\str{v},\setstr)|\leq \alpha=O\left(\epsilon^{-1}\ell\log \ell \log^2(n\ell/\beta)\right)$$
    \item Every string $P\in \Sigma^{[1,\ell]}$ which is not present $T^*$ was either: (i) not present in $\candidates$, in which case $\counts(P,\setstr)=O(\epsilon^{-1}\ell\log \ell (\log(n\ell /\beta)+\log |\Sigma|))$ by Lemma~\ref{lem:pruning}; or (ii) deleted in the pruning process of $\candidatetrie$, in which case $\counts(P,\setstr)<3\alpha=O(\epsilon^{-1}\ell\log \ell \log^2(n\ell/\beta))$.
    \item Any string $P\in\sum^{[1,\ell]}$ which is present in $T^*$ has a count at least $2\alpha-\alpha>1$. Therefore, $T^*$ has at most $O(n\ell^2)$ nodes.
\end{itemize}
To query the resulting data structure for a pattern $P$, we match $P$ in the trie, and if there is a node $v$ with $\str{v}=P$, we return the approximate count $\approxcount(\str{v},\setstr)$ saved at $v$. If $P$ is not present, we return $0$. This requires $O(|P|)$ time. 
In summary, we have obtained Theorem~\ref{thm:main_eps}.

\subsection{Faster Algorithm for Fixed-Length $q$-grams}

If we only care about counting $q$-grams for a fixed length $q$, the algorithm can be simplified significantly, improving the construction time significantly and the error slightly. 

\thmEpsQgrams*

\begin{proof}[Proofsketch]The algorithm constructs $\candidatespruned_{2^0},\dots, \candidatespruned_{2^j}$, where $j=\lfloor \log q\rfloor$, as in the proof of Lemma~\ref{lem:pruning}, with privacy parameter $\epsilon'=\epsilon/2$. We can then construct $\candidates_q$, again as in Lemma~\ref{lem:pruning}. We compute noisy counts of every string in the set $\candidates_q$ with the Laplace mechanism (Lemma~\ref{lem:Laplacemech}) and privacy parameter $\epsilon/2$. As a post-processing step, we throw out all elements from $\candidates_q$ whose noisy count is below $2\alpha$, where $\alpha$ is the error of the Laplace mechanism. We call the resulting set $\candidatespruned_q$ and arrange its elements into a trie together with the corresponding noisy counts. By composition, this is $\epsilon$-differentially private. By a similar argument as Lemma~\ref{lem:pruning}, with probability at least $1-\beta$, the noisy counts of any element in the trie has an error at most $O\left(\epsilon^{-1}\log \ell \left(\log(n\ell/\beta)+\log(|\Sigma|)\right)\right)$, and everything not in the trie has a count at most $O\left(\epsilon^{-1}\log \ell \left(\log(n\ell/\beta)+\log(|\Sigma|)\right)\right)$. Again by a similar argument as Lemma~\ref{lem:pruning}, with probability at least $1-\beta$, we have $|\candidatespruned_q|\leq n\ell$, and therefore the trie has size at most $n\ell^2$.

As shown in Step 1 of the proof of Lemma~\ref{lem:C_construction_algo}, $\candidatespruned_{2^0},\dots, \candidatespruned_{2^j}$ can be constructed using total time $O(n^2\ell^2\log q\log \log (n\ell))$ and $O(n\ell^2)$ space. The set $\candidates_q$ and the true count of every string in $\candidates_q$ can be computed in an additional $O(n^2\ell^2\log \log (n\ell)+n^2\ell^3+\mathrm{sort}(n\ell^2,|\Sigma|)) = O(n^2\ell^2\log \log (n\ell)+n^2\ell^3)$ time and $O(n^2\ell^3)$ space. The trie can be constructed in time $O(n^2\ell^3)$.
\end{proof}


\section{Counting with Approximate Differential Privacy}\label{sec:epsdel}
In this section, we prove Theorem~\ref{thm:main_epsdel}.

\mainEpsdelThm*
The algorithm proceeds in the same four main steps as the $\epsilon$-differentially private algorithm. To analyze this strategy for $(\epsilon,\delta)$-differential privacy, we need the following lemma:
\begin{lemma}\label{lem:l2norm}
    Let $v\in \mathbb{R}^k$ be a $k$-dimensional vector such that:
    \begin{itemize}
        \item $||v||_1=|v[1]|+\dots+|v[k]|\leq M$
        \item $|v[i]|\leq \Delta$, for all $i\in[1,k]$.
    \end{itemize}
    Then $||v||_2\leq \sqrt{M\Delta}$.
\end{lemma}
\begin{proof}
The proof follows directly from Hölder's inequality.
The inequality gives that for any $f,g\in \mathbb{R}^k$, and any $p,q\in[1,\infty]$ with $1/p+1/q=1$ it holds that $||f\cdot g||_1\leq ||f||_p||g||_q.$
    If we now set $p=\infty$ and $q=1$ and $f=g=v$, we get 
        $||v||_2^2=||v\cdot v||_1\leq ||v||_{\infty}||v||_1=M\Delta$
    and thus $||v||_2\leq \sqrt{M\Delta}$.
\end{proof}


This together with Corollary~\ref{cor:l1sens_same_length}  implies the following bound.
\begin{corollary}\label{cor:l2sens_same_length}
    For any $m\leq \ell$, the $L_2$-sensitivity of $(\counts_{\Delta}(P,\setstr))_{P\in \Sigma^m}$ is bounded by $\sqrt{2\ell\Delta}$. 
\end{corollary}


In the following, we give the details and analysis of our approach for $(\epsilon,\delta)$-differential privacy.

\subsection{Computing a Candidate Set}
We show how to compute the candidate set $\candidates$ while satisfying $(\epsilon,\delta)$-differential privacy:
\begin{lemma}\label{lem:pruning_epsdel}
    Let $\setstr=S_1,\dots, S_{n}$ be a collection of strings. For any $\epsilon>0$, $\delta>0$, and $0<\beta<1$, there exists an $(\epsilon,\delta)$-differentially private algorithm, which computes a candidate set $\candidates\subseteq\Sigma^{[1,\ell]}$ with the following two properties, with probability at least $1-\beta$:
    \begin{itemize}
        \item For any $P\in \Sigma^{[1,\ell]}$ not in $\candidates$, we have $$\counts_{\Delta}(P,\setstr)=O\left(\epsilon^{-1}\log \ell\sqrt{\ell\Delta\log(\log\ell/\delta)\log(\max\{\ell^2n^2,|\Sigma|\}/\beta)}\right).$$
        \item $|\candidates|\leq n^2\ell^3$.
    \end{itemize}
\end{lemma}
\begin{proof}
First, we inductively construct sets $\candidatespruned_{2^0},\dots,\candidatespruned_{2^j}$ for $j=\lfloor \log \ell\rfloor$, where $\candidatespruned_{2^k}$ contains only strings of length $2^k$ which have a sufficiently high count in $\setstr$. For every $k$, we use an $(\epsilon_1,\delta_1)$-differentially private algorithm to construct $\candidatespruned_{2^k}$ from $\candidatespruned_{2^{k-1}}$, where $\epsilon_1=\epsilon/(\lfloor \log \ell\rfloor +1)$ and $\delta_1=\delta/(\lfloor \log \ell\rfloor +1)$, such that the full algorithm fulfills $(\epsilon,\delta)$-differential privacy. In a post-processing step, we construct sets $\candidates_m$ of strings of length $m$ for every $1\leq m\leq \ell$, such that every pattern of length $m$ which is not in $\candidates_m$ has a small count with high probability. We define $\candidates$ as the union of all $\candidates_m$.

{\bf Computing $\candidatespruned_{2^0}$.} First, we estimate $\counts_{\Delta}(\gamma,\setstr)$ of all letters $\gamma\in \Sigma$. By Corollary~\ref{cor:l2sens_same_length}, the sensitivity of $(\counts_{\Delta}(\gamma,\setstr))_{\gamma\in \Sigma}$ is bounded by $\sqrt{2\ell\Delta}$. Let $\beta_1=\beta/(\lfloor \log \ell\rfloor +1)$. Using the algorithm given by Corollary~\ref{cor:gaussian_mech}, we compute an estimate $\approxcount_{\Delta}(\gamma,\setstr)$ such that
$$\max_{\gamma\in\Sigma}|\approxcount_{\Delta}(\gamma,\setstr)-\counts_{\Delta}(\gamma,\setstr)|\leq2\epsilon_1^{-1}\sqrt{2\ell\Delta\ln(2/\delta_1)\ln(2|\Sigma|/\beta_1)}$$
with probability at least $1-\beta_1$, while satisfying $(\epsilon_1,\delta_1)$-differential privacy. In the following, let $\alpha=2\epsilon_1^{-1}\sqrt{2\ell\Delta\ln(2/\delta_1)\ln(2\max\{\ell^2n^2,|\Sigma|\}/\beta_1)}$.
We keep a pruned candidate set $\candidatespruned_{2^0}$ of strings of length 1 with an approximate count at least $\tau=2\alpha$, i.e., we keep all $\gamma$ with $\approxcount_{\Delta}(\gamma,\setstr)\geq \tau$. If $|\candidatespruned_{2^0}|>n\ell$, we stop the algorithm and return a fail message. 

{\bf Computing $\candidatespruned_{2^k}$ for $k=1,\dots,\lfloor \log \ell \rfloor$.} Given $\candidatespruned_{2^{k-1}}$ with $|\candidatespruned_{2^{k-1}}|\leq n\ell$ and $k\geq 1$, compute $\candidatespruned_{2^k}$ as follows: first, construct the set $\candidatespruned_{2^{k-1}}\circ\candidatespruned_{2^{k-1}}$, i.e., all strings that are a concatenation of two strings in $\candidatespruned_{2^{k-1}}$. There are at most $(n\ell)^2$ of these.  Again by Corollary~\ref{cor:l2sens_same_length}, the $L_2$-sensitivity of $\counts_{\Delta}$ for all strings in $\candidatespruned_{2^{k-1}}\circ \candidatespruned_{2^{k-1}}$ is bounded by $\sqrt{2\ell\Delta}$.  We use the algorithm from  Corollary~\ref{cor:gaussian_mech} to estimate $\counts_{\Delta}$ for all strings in $\candidatespruned_{2^{k-1}}\circ \candidatespruned_{2^{k-1}}$ up to an additive error at most  $2\epsilon_1^{-1}\sqrt{2\ell\Delta\ln(2/\delta_1)\ln(2\ell^2n^2/\beta_1)}\leq \alpha$ with probability at least $1-\beta_1$ and $(\epsilon_1,\delta_1)$-differential privacy. The set $\candidatespruned_{2^k}$ is the set of all strings in $\candidatespruned_{2^{k-1}}\circ \candidatespruned_{2^{k-1}}$ with an approximate count at least $\thresh$. Again, if $|\candidatespruned_{2^k}|> n\ell$, we stop the algorithm and return a fail message. 

{\bf Constructing $\candidates$.} From $\candidatespruned_{2^0},\dots,\candidatespruned_{2^j}$, we now construct, for each $m$ which is not a power of two, a set of candidate patterns $\candidates_m$, without taking $\setstr$ further into account. Specifically, for any fixed length $m$ with $2^k<m<2^{k+1}$, for $k=0,\dots, \lfloor \log \ell \rfloor$, we define $\candidates_m$ as the set of all strings $P$ of length $m$ such that $P[0,2^{k}-1]\in \candidatespruned_{2^k}$ and $P[m-2^{k},m-1]\in \candidatespruned_{2^k}$. For $m=2^k$ for some $k\in \{0,\dots, \lfloor \log \ell \rfloor\}$, we define $\candidates_m=\candidatespruned_{2^k}$. We define $\candidates=\bigcup_{m=1}^{\ell}\candidates_m$. 

{\bf Privacy.} Since there are $\lfloor\log \ell\rfloor +1$ choices of $j$, and for each we run an $(\epsilon_1,\delta_1)$-differentially private algorithm for $\epsilon=\epsilon/(\lfloor\log \ell\rfloor +1)$ and $\delta_1=\delta/(\lfloor\log \ell\rfloor +1)$, their composition is $(\epsilon,\delta)$-differentially private by Lemma~\ref{lem:composition_theorem}. Since constructing $\candidates$ from $\candidatespruned_{2^0},\dots,\candidatespruned_{2^j}$ is post-processing, the entire algorithm is $(\epsilon,\delta)$-differentially private.

{\bf Accuracy.} Since there are $\lfloor\log \ell\rfloor +1$ choices of $k$, by the choice of $\beta_1$, and by the union bound, all the error bounds hold together with probability at least $1-\beta$.  Let $E$ be the event that all the error bounds holding simultaneously. In the following, we conditioned on $E$.  
Conditioned on $E$, for any $P\in \candidatespruned_{2^k}$, $k=0,\dots,\lfloor \log \ell \rfloor$, we have that $\counts_{\Delta}(P,\setstr)\geq \thresh-\alpha\geq \alpha>1$, i.e., it appears at least once as a substring in $\setstr$. Note that any string in $\setstr$ has at most $\ell$ substrings of length $2^k$. Since there are $n$ strings in $\setstr$, we have $|\candidatespruned_{2^k}|\leq n\ell$, conditioned on $E$. Thus, conditioned on $E$, the algorithm does not abort. Additionally, any $P$ of length $2^k$ which is not in $\candidatespruned_{2^k}$ satisfies $\counts_{\Delta}(P,\setstr)<\thresh+\alpha=3\alpha$. 
  Now, any pattern $P$ of length $2^k<m<2^{k+1}$ for some $k\in\{0,\dots,\lfloor\log \ell \rfloor\}$ satisfies that $\counts_{\Delta}(P,\setstr)\leq \counts_{\Delta}(P[0,2^k-1],\setstr)$ and $\counts_{\Delta}(P,\setstr)\leq \counts_{\Delta}(P[m-2^k,m-1],\setstr)$. Since $P\notin\candidates_m$ if and only if either $P[0,2^k-1]\notin \candidatespruned_{2^k}$ or $P[m-2^k
  ,m-1]\notin \candidatespruned_{2^k}$, we have that $P\notin\candidates_m$ implies $\counts_{\Delta}(P,\setstr)<3\alpha$, conditioned on $E$. As $P[0,2^k-1]$ and $P[m-2^k,m-1]$ cover $P$ completely, we have $|\candidates_m|\leq |\candidatespruned_{2^k}|^2\leq (n\ell)^2$. Therefore, $|\candidates|=\sum_{m=1}^{\ell}|\candidates_m|\leq n^2\ell^3$.  
    \end{proof}

\subsection{Heavy Path Decomposition and Properties}
In the following, let $\candidatetrie$ be the trie of a set $\candidates$ fulfilling the properties of Lemma~\ref{lem:pruning_epsdel}. As a next step, we build the heavy path decomposition for $\candidatetrie$, as defined in Section~\ref{subsec:heavypath}. We redefine the difference sequence on a path for $\counts_{\Delta}$ (in Section~\ref{subsec:heavypath} we only defined it for $\counts$).

    \paragraph{Difference sequence} For any path $p=v_0,v_1,\dots, v_{|p|-1}$ in the trie $\candidatetrie$, the \emph{difference sequence of $\counts_{\Delta}$ on $p$} is given by the $|p|-1$ dimensional vector $\diff_p(\setstr)[i]=\counts_{\Delta}(\str{v_i},\setstr)-\counts_{\Delta}(\str{v_{i-1}},\setstr)$ for $i=1\dots |p|-1$.

 We need the following generalizations of Observation~\ref{obs:sensitvity_node} and Lemma~\ref{lem:sensitvity_path}: 
\begin{lemma}\label{lem:sensitivity_tree_Del}
 Let $\setstr$ and $\setstr'$ be neighboring data sets, such that $\setstr' = \setstr\setminus\{S\}\cup\{S'\}$.
\begin{enumerate}
    \item \label{lem:sensitvity_node_Del} For any node $v$ in the trie $\candidatetrie$, we have $$|\counts_{\Delta}(\str{v}, \setstr)-\counts_{\Delta}(\str{v}, \setstr')|= |\counts_{\Delta}(\str{v},S)- \counts_{\Delta}(\str{v},S')|.$$
    \item \label{lem:sensitvity_path_Del} For any path $p$ with root $r$ in the trie $\candidatetrie$, we have $$||\diff_p(\setstr)-\diff_p(\setstr')||_1\leq \counts_{\Delta}(\str{r},S)+ \counts_{\Delta}(\str{v},S').$$
\end{enumerate}
\end{lemma}
\begin{proof} Item~\ref{lem:sensitvity_node_Del} follows from the definition of $\counts_{\Delta}$. For item~\ref{lem:sensitvity_path_Del}, note that if $v'$ is a descendant of $v$, then $\str{v}$ is a prefix of $\str{v'}$ and therefore $\counts_ {\Delta}(\str{v},S)\geq \counts_ {\Delta}(\str{v'},S)$. Now item~\ref{lem:sensitvity_path_Del} follows from the fact that $\counts_{\Delta}(\str{v_i},S)$ is monotonically non-increasing in $i$ for $p=v_0, v_1,\dots, v_{|p|-1}$ with $v_0=r$, and from $0\leq \counts_{\Delta}(\str{v_i},S)\leq \counts(\str{r},S)$, for any descendant $v_i$ of $r$. The same argument applies for $S'$ and thus item~\ref{lem:sensitvity_path_Del} follows.
\end{proof}
\begin{lemma}\label{lem:all_roots_Del}
    Let $r_0,r_1,\dots, r_k$ be the roots of the paths of the heavy path decomposition of $\candidatetrie$, where $r_0$ is the root of $\candidatetrie$. Then  for any string $S$ of length at most $\ell$, we have $\sum_{i=0}^k \counts_{\Delta}(\str{r_i},S))\leq \ell(\lfloor\log |\candidatetrie|\rfloor+1)=O( \ell \log (n\ell))$.
\end{lemma}
\begin{proof}
    This lemma is a direct Corollary of Lemma~\ref{lem:all_roots}.
\end{proof}

By Lemma~\ref{lem:sensitivity_tree_Del}.\ref{lem:sensitvity_node_Del} and Lemma~\ref{lem:all_roots_Del}, the $L_1$-sensitivity of $\left(\counts_{\Delta}(\str{r_i},\setstr)\right)_{i=0}^k$ is bounded by $2\ell(\lfloor\log |\candidatetrie|\rfloor +1)=O(\ell \log (n\ell))$. Since also $\counts_{\Delta}(\str{r_i},S)\leq \Delta$ for every $i=0,\dots,k$ and any string $S$, we have that the $L_2$-sensitivity of $\left(\counts_{\Delta}(\str{r_i},\setstr)\right)_{i=0}^k$ is bounded by $O(\sqrt{\Delta\ell \log (n\ell)})$ by Lemma~\ref{lem:l2norm}. Additionally, note that since $\candidatetrie$ has at most $n^2\ell^3$ leaves, there are at most $n^2\ell^3$ heavy paths. Together with Corollary~\ref{cor:gaussian_mech}, this gives:
\begin{corollary}\label{cor:heavypathroots_epsdel}
Let $r_0,\dots, r_k$ be the roots of the heavy paths of $\candidatetrie$. For any $\epsilon>0$, $\delta>0$, and $0<\beta<1$, 
    there exists an $(\epsilon,\delta)$-differentially private algorithm, which estimates $\counts_{\Delta}(\str{r_i}, \setstr)$ for the heavy path roots $r_0,\dots, r_k$ up to an additive error at most $O(\epsilon^{-1}\sqrt{\Delta \ell\log(n\ell)\ln(1/\delta)\ln(k/\beta)})=O(\epsilon^{-1}\sqrt{\Delta \ell\ln(1/\delta)}\ln(n\ell/\beta))$ with probability at least $1-\beta$.
\end{corollary}
Additionally, Lemma~\ref{lem:all_roots_Del} and Lemma~\ref{lem:sensitivity_tree_Del}.\ref{lem:sensitvity_path_Del} give that the \emph{sum of} $L_1$-sensitivities of the difference sequences \emph{over all heavy paths} is bounded by $2\ell(\lfloor\log |\candidatetrie|\rfloor +1)$, and Lemma~\ref{lem:sensitivity_tree_Del}.\ref{lem:sensitvity_path_Del} gives that the $L_1$-sensitivity of the difference sequence on \emph{one} heavy path is bounded by $\Delta$. As a next step, we show how we can compute all \emph{prefix sums} of the difference sequences of all heavy paths with an error roughly $O(\sqrt{\ell(\lfloor\log |\candidatetrie|\rfloor+1)\Delta})$ (up to poly-logarithmic terms). To do this, we extend Lemma~\ref{lem:binary_trees} to $(\epsilon,\delta)$-differential privacy.

\begin{lemma}\label{lem:binary_trees_epsdel}
Let $\chi^{*}$ be any universe of possible data sets. Let $a^{(1)}, \dots, a^{(k)}$ be $k$ functions, where the output of every function is a sequence of length $T$, i.e. $a^{(i)}:\chi^{*}\rightarrow \mathbb{N}^{T}$ for all $i\in[1,k]$. If $a^{(1)}, \dots, a^{(k)}$ fulfill the following properties for any two neighboring $x$ and $x'$ from $\chi^{*}$:
\begin{itemize}
    \item $\sum_{j=1}^T |a^{(i)}(x)[j]-a^{(i)}(x')[j]|\leq \Delta$ for all $i=1,\dots, k$,
    \item $\sum_{i=1}^k\sum_{j=1}^T |a^{(i)}(x)[j]-a^{(i)}(x')[j]|\leq L$,
\end{itemize} 
then for any $\epsilon>0$, any $\delta>0$, and $0<\beta<1$, there exists an $(\epsilon,\delta)$-differentially private algorithm computing for every $i\in[1,k]$ all prefix sums of $a^{(i)}(x)$ with additive error at most $O(\epsilon^{-1}\sqrt{L\Delta\ln(1/\delta)\log(Tk/\beta)}\log T)$ with probability at least $1-\beta$.
\end{lemma}
For the proof, we need the following fact on the sum of Normal Distributions.

\begin{fact}[Sum of Normal Distributions]\label{fact:sum_of_gaussians} Let $X_1\approx N(0,\sigma_1^2)$ and $X_2\approx N(0,\sigma_2^2)$ be independently drawn random variables. Then $Y=X_1+X_2$ fulfills $Y\approx N(0, \sigma_1^2+\sigma_2^2)$.
\end{fact}

\begin{proof}[Proof of Lemma~\ref{lem:binary_trees_epsdel}]

The algorithm builds a copy of the binary tree mechanism \cite{Dwork2010} for each of the $k$ sequences. 
The dyadic decomposition of an interval $[1,T]$ is given by the set $\mathcal{I}=\{[j\cdot 2^i+1, (j+1)2^{i}],0\leq j\leq \lceil T/2^i\rceil -1, 0\leq i \leq \lfloor \log T \rfloor\}$. It has the property that every interval $[1, m]$ for $m\leq T$ is the union of no more than $\lfloor\log T\rfloor+1$ disjoint intervals from $\mathcal{I}$. Denote these intervals by $\mathcal{I}_{[1,m]}$.  The algorithm is now as follows: For every $i$ in $[1,k]$ and interval $[b,e]\in\mathcal{I}$, we independently draw a random variable $Y_{[b,e]}^{(i)}\approx N(0, \sigma^2)$, where $\sigma=\epsilon^{-1}\sqrt{2L\Delta (\lfloor \log T \rfloor +1)\ln(2/\delta)}$. Given the sequences $a^{(1)}, \dots, a^{(k)}$ and the data set $x$ we compute for every sequence $a^{(i)}(x)$, and every interval $[b,e]\in\mathcal{I}$, the partial sum $s^{(i)}_{[b,e]}(x)=a^{(i)}(x)[b]+\dots+a^{(i)}(x)[e]$. For each such partial sum, we compute an approximate sum $\tilde{s}^{(i)}_{[b,e]}=s^{(i)}_{[b,e]}(x)+Y_{[b,e]}^{(i)}$. We now compute the approximate $m$th prefix sum of the $i$th sequence by $\sum_{[b,e]\in \mathcal{I}_{[1,m]}} \tilde{s}^{(i)}_{[b,e]}$, for all $i=1,\dots, k$ and all $m=1,\dots, T$.

{\bf Privacy.} We show that the $L_2$-sensitivity of the sequence of $s^{(i)}_{[b,e]}(x)$  with $i=1,\dots k$ and $[b,e]\in \mathcal{I}$ is at most $\sqrt{L\Delta(\lfloor \log T \rfloor +1)}$: Recall that $L\geq\max_{x\sim x'}\sum_{i=1}^k \sum_{j=1}^{T}|a^{(i)}(x)[j]-a^{(i)}(x')[j]|$. Thus, $L_1$-sensitivity for all $s^{(i)}_{[b,e]}(x)$ is bounded by 
\begin{align*}
    \max_{x\sim x'}\sum_{i=1}^k\sum_{[b,e]\in \mathcal{I}}|s^{(i)}_{[b,e]}(x)-s^{(i)}_{[b,e]}(x')|&\leq  \max_{x\sim x'}\sum_{i=1}^k\sum_{[b,e]\in \mathcal{I}}\sum_{j\in [b,e]}|a^{(i)}(x)[j]-a^{(i)}(x')[j]|\\&=\max_{x\sim x'}\sum_{i=1}^k\sum_{j=1}^{T}\sum_{[b,e]\in \mathcal{I}:j\in [b,e]}|a^{(i)}(x)[j]-a^{(i)}(x')[j]|\\&\leq L(\lfloor \log T \rfloor +1),
\end{align*}
where the last inequality follows from the fact that any $j\in[1,T]$ appears in no more than $\lfloor \log T \rfloor +1$ many sets in $\mathcal{I}$. 
Further, since $\sum_{j=1}^T |a^{(i)}(x)[j]-a^{(i)}(x')[j]|\leq \Delta$ for all $i=1,\dots, k$, we have in particular $|s^{(i)}_{[b,e]}(x)-s^{(i)}_{[b,e]}(x')|\leq \Delta$ for all $[b,e]\in \mathcal{I}$ and all $i=1,\dots,k$. Lemma~\ref{lem:l2norm} gives that the $L_2$-sensitivity for all $s^{(i)}_{[b,e]}(x)$ is bounded by $\sqrt{\Delta L(\lfloor \log T \rfloor +1))}$.
Thus, the algorithm above is $(\epsilon,\delta)$-differentially private by Lemma~\ref{lem:gaussianmech}. 

{\bf Accuracy.} 
For any $i$ and $m$, the estimate of the $m$th prefix sum is given by $\sum_{I\in\mathcal{I}_{[1,m]}}\tilde{s}^{(i)}_I=\sum_{j=1}^m a^{(i)}(x)[j]+ \sum_{I\in\mathcal{I}_{[1,m]}}Y^{(i)}_I$. Note that $Y_{[1,m]}^{(i)}= \sum_{I\in\mathcal{I}_{[1,m]}}Y^{(i)}_I$ is the sum of at most $\lfloor \log T\rfloor +1$ independent random variables with distribution $N(0,\sigma^2)$ with $\sigma=\epsilon^{-1}\sqrt{2L\Delta(\lfloor \log T\rfloor +1)\ln(2/\delta)}$. Thus, $Y_{[1,m]}^{(i)}\approx N(0,\sigma_1^2)$, where $\sigma_1^2\leq (\lfloor \log T\rfloor +1)\sigma^2$, by Fact~\ref{fact:sum_of_gaussians}. By the Gaussian tail bound (see Lemma~\ref{lem:gaussian_tail}) we have 
        \begin{align*}
            \Pr[|Y_{[1,m]}^{(i)}|\geq \sigma_1\sqrt{\log(Tk/\beta)}]\leq \frac{\beta}{Tk}.
        \end{align*}
        Thus, with probability at least $1-\beta$, the error on all prefix sums is bounded by 
        \begin{align*}
            \sigma_1\sqrt{\log(Tk/\beta)}&=\sigma\sqrt{(\lfloor \log T\rfloor +1)\log(Tk/\beta)}\\&=\epsilon^{-1}\sqrt{2L\Delta (\lfloor \log T\rfloor +1)\ln(2/\delta)}\sqrt{(\lfloor \log T\rfloor +1)\log(Tk/\beta)}\\& = O(\epsilon^{-1}\sqrt{L\Delta\ln(1/\delta)\log(Tk/\beta)}\log T).
        \end{align*}
 
\end{proof}

Lemma~\ref{lem:binary_trees_epsdel} now gives the following corollary with $L=2\ell(\lfloor\log |\candidatetrie|\rfloor +1)=O(\ell \log(n\ell))$, $k\leq n^2\ell^3$, and $T=\ell$. 
\begin{corollary}\label{cor:heavypaths_epsdel}
    Let $p_0,\dots,p_k$ be all heavy paths of $\candidatetrie$. For any $\epsilon>0$, $\delta>0$, and $0<\beta<1$, there exists an $(\epsilon,\delta)$-differentially private algorithm, which estimates $\sum_{j=1}^i \diff_{p_m}(\setstr)[j]$ for all $i=1,\dots,|p_m|$ and all $m=0,\dots,k$ up to an additive error at most 
    \begin{align*}
        O\left(\epsilon^{-1}\sqrt{\Delta\ell \log(n\ell)\ln(1/\delta)\log(\ell k/\beta)}\log \ell\right)=O\left(\epsilon^{-1}\sqrt{\Delta\ell\ln(1/\delta)}\log(n\ell/\beta)\log \ell\right).
    \end{align*}
\end{corollary}
\subsection{Putting It Together}
Our full algorithm runs the algorithms given by Lemma~\ref{lem:pruning_epsdel}, Corollary~\ref{cor:heavypathroots_epsdel} and Corollary~\ref{cor:heavypaths_epsdel} in sequence with privacy parameters $\epsilon'=\epsilon/3$, $\delta'=\delta/3$, and failure probability $\beta'=\beta/3$. For a heavy path $p= v_0, v_1,\dots, v_{|p|-1}$ with root $r=v_0$, let $\approxcount_{\Delta}(\str{r},\setstr)$ be the approximate counts computed by the algorithm from Corollary~\ref{cor:heavypathroots_epsdel}, and let $\mathrm{sums}^*_p[i]$ be the approximate estimate of $\sum_{j=1}^i \diff_p(\setstr)[j]$. Then for any node $v_i$ on the path with $i>0$, we compute $\approxcount_{\Delta}(\str{v_i},\setstr)=\approxcount_{\Delta}(\str{r},\setstr)+\mathrm{sums}^*_p[i]$. Let $\alpha$ be the sum of the error bounds for Corollary~\ref{cor:heavypathroots_epsdel} and Corollary~\ref{cor:heavypaths_epsdel} which each hold with probability $1-\beta'$. We have $\alpha=O\left(\epsilon^{-1}\sqrt{\Delta\ell\ln(1/\delta)}\log(n\ell/\beta)\log \ell\right)$. As a final step, we prune $\candidatetrie$ by traversing the trie from the root, and for any node $v$ with $\approxcount_{\Delta}(\str{v_i},\setstr)<2\alpha$, we delete $v$ and its subtree. We return the resulting pruned trie $T^*$ together with approximate counts $\approxcount_{\Delta}(\str{v},\setstr)$ for every node $v 
\in T^*$.

By the composition theorem (Lemma~\ref{lem:composition_theorem}), this algorithm is $(\epsilon,\delta)$-differentially private. By a union bound, all its error guarantees hold together with probability $1-3\beta'=1-\beta$. Thus, we get with probability $1-\beta$ a trie $T^*$, with the following properties:
\begin{itemize}
    \item  By Corollaries~\ref{cor:heavypathroots_epsdel} and \ref{cor:heavypaths_epsdel}, for each node $v\in T^*$
    $$|\approxcount_{\Delta}(\str{v},\setstr)-\counts_{\Delta}(\str{v},\setstr)|\leq \alpha=O\left(\epsilon^{-1}\sqrt{\Delta\ell\ln(1/\delta)}\log(n\ell/\beta)\log \ell\right)$$

    \item Every string $P\in \Sigma^{[1,\ell]}$ which is not present $T^*$ was either: (i) not present in $\candidates$, in which case $$\counts_{\Delta}(P,\setstr)=O\left(\epsilon^{-1}\log \ell\sqrt{\ell\Delta\log(\log\ell/\delta)\log(\max\{\ell^2n^2,|\Sigma|\}/\beta)}\right)$$ by Lemma~\ref{lem:pruning_epsdel}; or (ii) deleted in the pruning process of $\candidatetrie$, in which case $$\counts_{\Delta}(P,\setstr)<3\alpha=O\left(\epsilon^{-1}\sqrt{\Delta\ell\ln(1/\delta)}\log(n\ell/\beta)\log \ell\right).$$
    \item Any string $P\in\sum^{[1,\ell]}$ which is present in $T^*$ has a count at least $2\alpha-\alpha>1$. Therefore, $T^*$ has at most $O(n\ell^2)$ nodes.
\end{itemize}
To query the resulting data structure for a pattern $P$, we match $P$ in the trie, and if there is a node $v$ with $\str{v}=P$, we return the approximate count $\approxcount_{\Delta}(\str{v},\setstr)$ saved at $v$. If $P$ is not present, we return $0$. This requires $O(|P|)$ time. 
Together, this gives Theorem~\ref{thm:main_epsdel}.

\subsection{Faster Algorithm for Fixed-Length $q$-grams}\label{sec:EpsilonDeltaQgrams}

In this section, we give an algorithm for $(\epsilon,\delta)$-differentially private $q$-grams, which can be made to run in $O(n\ell(\log q + \log|\Sigma|))$ time and $O(n\ell)$ space. 
The idea is that for this less stringent version of approximate differential privacy, we can avoid computing the noisy counts of substrings that do not appear in the database, by showing that with high probability, these strings will not be contained in the result, anyway. Specifically, we will use the following lemma.
\begin{lemma}\label{lem:magic}
    Let $\alg_1:\chi^*\rightarrow \mathcal{R}$ and $\alg_2:\chi^*\rightarrow \mathcal{R}$ be two randomized algorithms. Let $\epsilon>0$, $\delta\in[0,1)$ and $\gamma\leq\frac{\delta}{3e^{\epsilon}}$. Let $\alg_1$ be $(\epsilon,\delta')$-differentially private, where $\delta'\leq \gamma$.
    Assume that for every data set $\setstr\in \chi^*$, there exists an event $E(\setstr)\subseteq \mathcal{R}$ such that for all $U\subseteq\mathcal{R}$
    \[\Pr[\alg_1(\setstr)\in U|\alg_1(\setstr)\in E(\setstr)]=\Pr[\alg_2(\setstr)\in U]\] and $\Pr[\alg_1(\setstr)\in E(\setstr)]\geq1-\gamma$. Then $\alg_2$ is $(\epsilon,\delta)$-differentially private.
\end{lemma}
\begin{proof}
    Let $\setstr$ and $\setstr'$ be neighbouring data sets and let $(\alg_1(\setstr)\in E(\setstr))=E$ and $(\alg_1(\setstr')\in E(\setstr'))=E'$. We have, for any $U\in\mathcal{R}$:
    \begin{align*}
        \Pr[\alg_2(\setstr)\in U]&=\Pr[\alg_1(\setstr)\in U|E]
        \leq \frac{\Pr[\alg_1(\setstr)\in U]}{1-\gamma}
        \leq \frac{e^{\epsilon}\Pr[\alg_1(\setstr')\in U]+\delta'}{1-\gamma}\\
        &\leq \frac{e^{\epsilon}(\Pr[\alg_1(\setstr')\in U|E'](1-\gamma)+\gamma)+\delta'}{1-\gamma}\\
        &\leq e^{\epsilon} \Pr[\alg_1(\setstr')\in U|E']+\frac{e^{\epsilon}\gamma+\delta'}{1-\gamma}\\
        &\leq e^{\epsilon} \Pr[\alg_1(\setstr')\in U|E']+\delta\\
        &= e^{\epsilon} \Pr[\alg_2(\setstr')\in U]+\delta
    \end{align*}
    The third inequality holds since $\Pr[E']\geq 1-\gamma$ and the last inequality holds since
    \begin{align*}
    \frac{e^{\epsilon}\gamma+\delta'}{1-\gamma}\leq \frac{2e^{\epsilon}\gamma}{1-\gamma}.
    \end{align*}
    Setting $\frac{2e^{\epsilon}\gamma}{1-\gamma}\leq \delta$ gives
    \(  (2e^{\epsilon}+\delta)\gamma\leq \delta \)
    and therefore it is enough to choose
     \(  \gamma\leq \frac{\delta}{3e^{\epsilon}}\leq \frac{\delta}{2e^{\epsilon}+\delta}.\)
\end{proof}

\begin{lemma}\label{lem:qgram_eps_delta}
 Let $n,\ell$, and $q\leq\ell$ be integers and $\Sigma$ an alphabet of size $|\Sigma|$. Let $\Delta\leq \ell$. For any $\epsilon>0$, $\delta>0$, and $0<\beta<1$, there exists an $(\epsilon,\delta)$-differentially private algorithm, which can process any database $\setstr=S_1,\dots,S_{n}$ of documents from $\Sigma^{[1,\ell]}$ and with probability at least $1-\beta$ output a data structure for 
 $\counts_{\Delta}$ for all $q$-grams with additive error $~O\left(\epsilon^{-1}\sqrt{\ell\Delta\log (n\ell)}\log q\left(\epsilon+\log\log q+ \log \frac{|\Sigma|}{\delta\beta}\right)\right)$. 
\end{lemma}

\begin{proof}
The idea is to design two algorithms, $\alg_1$ and $\alg_2$, where $\alg_2$ is the algorithm that we will run, and $\alg_1$ is an algorithm that we show to be differentially private with appropriate parameters. We then argue that the two algorithms behave the same, conditioned on a high-probability event.
We set $\epsilon_1=\epsilon/(\lfloor \log q\rfloor +2)$, $\beta_1=\min(\beta/(\lfloor \log q\rfloor +2),\delta/(3e^{\epsilon}(\lfloor \log q\rfloor +2))$, $\delta_1\leq \beta_1$.  Let $j=\lfloor \log q \rfloor$ and $\alg_1=\alg_1^{(j+1)}\circ\dots \circ\alg_1^{(0)}$.
    \begin{itemize}
        \item $\alg_1^{(0)}$ computes $\candidatespruned_{2^0}$ as follows: it adds noise scaled with $N(0,\sigma^2)$ to the count of each letter, where $\sigma=2\epsilon_1^{-1}\sqrt{2\ell\Delta\ln(2/\delta_1)}$. Let $\alpha=2\epsilon_1^{-1}\sqrt{2\ell\Delta\ln(2/\delta_1)\ln(2\max\{\ell^2n^2,|\Sigma|\}/\beta_1)}$. The set $\candidatespruned_{2^0}$ is the set of all letters with a noisy count at least $2\alpha$. We stop if $|\candidatespruned_{2^0}|>n\ell$.
        \item $\alg_1^{(k)}$, for $0<k<j+1$, computes $\candidatespruned_{2^k}$ by adding noise scaled with $N(0,\sigma^2)$ to the count of every string in $\candidatespruned_{2^{k-1}}\circ \candidatespruned_{2^{k-1}}$. $\candidatespruned_{2^k}$ is the set of all strings of length $2^k$ with noisy count at least $2\alpha$. We stop if $|\candidatespruned_{2^k}|>n\ell$. 
        \item $\alg_1^{(j+1)}$ computes $\candidatespruned_q$ by adding noise scaled with $N(0,\sigma^2)$ to the count of every string $P$ of length $q$ such that $P[0\dots 2^j-1]\in \candidatespruned_{2^j}$ and $P[q-2^j\dots q-1]\in \candidatespruned_{2^j}$. $\candidatespruned_q$ is the set of all such strings with a noisy count at least $2\alpha$ and we output it together with the noisy counts.
    \end{itemize}
    
    Let $\mathcal{Z}_m(\setstr)$ be the set of all strings $P$ of length $m$ which satisfy $\counts_{\Delta}(P,\setstr)=0$. $\alg_2$ is $\alg_2=\alg_2^{(j+1)}\circ\dots \circ\alg_2^{(0)}$, where for every $0\leq k \leq j+1$, the algorithm $\alg_2^{(k)}$ is the same as $\alg_1^{(k)}$, except it does not take into account strings whose true count is zero. In detail:  
      \begin{itemize}
        \item $\alg_2^{(0)}$ constructs $\candidatespruned_{2^0}$ by adding noise scaled with $N(0,\sigma^2)$ to the count of every letter in $\Sigma\setminus \mathcal{Z}_1(\setstr)$, where $\sigma=2\epsilon_1^{-1}\sqrt{2\ell\Delta\ln(2/\delta_1)}$. Let $\alpha=2\epsilon_1^{-1}\sqrt{2\ell\Delta\ln(2/\delta_1)\ln(2\max\{\ell^2n^2,|\Sigma|\}/\beta_1)}$: $\candidatespruned_{2^0}$ is the set of all letters with a noisy count at least $2\alpha$.
        \item $\alg_2^{(k)}$, for $0<k<j+1$, constructs $\candidatespruned_{2^k}$ by adding noise scaled with $N(0,\sigma^2)$ to the count of every string in $\left(\candidatespruned_{2^{k-1}}\circ \candidatespruned_{2^{k-1}}\right)\setminus \mathcal{Z}_{2^k}(\setstr)$. The set $\candidatespruned_{2^k}$ contains all strings with noisy count at least $2\alpha$. 
        \item $\alg_2^{(j+1)}$ constructs $\candidatespruned_q$ by adding noise scaled with $N(0,\sigma^2)$ to the count of every string $P\in\Sigma^q\setminus\mathcal{Z}_{q}(\setstr) $ of length $q$ such that $P[0\dots 2^j-1]\in \candidatespruned_{2^j}$ and $P[q-2^j\dots q-1]\in \candidatespruned_{2^j}$. $\candidatespruned_q$ is the set of all such strings with a noisy count at least $2\alpha$ together with the noisy counts.
    \end{itemize}

\textbf{Privacy of $\alg_2$.} 
Every $\alg_1^{(k)}$ for $0\leq k \leq j+1$ is $(\epsilon_1,\delta_1)$-differentially private by Corollary~\ref{cor:l2sens_same_length} and Lemma~\ref{lem:gaussianmech}.
  We use Lemma~\ref{lem:magic} on $\alg_1^{(k)}$ and $\alg_2^{(k)}$, for every $0\leq k \leq j+1$. We define $E^{(k)}(\setstr)$ as the event that none of the noises added by $\alg_1^{(k)}$ to strings in $\mathcal{Z}_k$ exceeds an absolute value of $\alpha$. By Lemma~\ref{lem:gaussian_tail}, this is true with probability at least $1-\beta_1$. Let $\mathcal{R}^{(k)}$ be the range of $\alg_1^{(k)}$ and $\alg_2^{(k)}$. 
  Now, for any $U\subseteq \mathcal{R}^{(k)}$ and since all the added noises are independent, $\Pr[\alg_1^{(k)}\in U|E^{(k)}(\setstr)]$ only depends on the distribution of the noises added to the counts of strings in $(\candidatespruned_{2^{k-1}}\circ \candidatespruned_{2^{k-1}})\setminus \mathcal{Z}_{2^k}(\setstr)$ and is thus equal to $\Pr[\alg_2^{(k)}\in U]$. 
 Since $\delta_1\leq \beta_1$ and $\beta_1\leq \frac{\delta}{3e^{\epsilon}(\lfloor \log q \rfloor +2)}$, by Lemma~\ref{lem:magic}, we have that $\alg_2^{(k)}$ is $(\epsilon_1,\delta/(\lfloor \log q \rfloor +2))$ differentially private. By Lemma~\ref{lem:composition_theorem}, we have that $\alg_2$ is $(\epsilon,\delta)$-differentially private.
 
 \textbf{Accuracy.} The accuracy proof proceeds similarly to the proof of Lemma~\ref{lem:pruning} (or Lemma~\ref{lem:pruning_epsdel}, which is its equivalent for $(\epsilon,\delta)$-differential privacy). By the same arguments, with probability at least $1-\beta$, every pattern of length $q$ not in $\candidatespruned_q$ has a count at most $3\alpha$, and we compute the counts of patterns in $\candidatespruned_q$ with error at most $\alpha$, where $\alpha=2\epsilon_1^{-1}\sqrt{2\ell\Delta\ln(2/\delta_1)\ln(2\max\{\ell^2n^2,|\Sigma|\}/\beta_1)}=O(\epsilon^{-1}\sqrt{\ell\Delta\log (n\ell)}\log q(\epsilon+\log\log q+ \log \frac{|\Sigma|}{\delta\beta}))$.
\end{proof}

Lemma~\ref{lem:qgram_eps_del_construct} provides an efficient procedure to construct sets $\candidatespruned_{2^0},\ldots, \candidatespruned_{2^j}$and $\candidatespruned_{q}$ for $\alg_2$. The efficiency of this procedure relies on the fact that
$\alg_2$ ignores all substrings that do not occur in any document of $\setstr$, in contrast with the algorithms underlying Theorem~\ref{thm:main_eps}. This crucial difference allows us to compute (noisy) counts only for carefully selected substrings of $\setstr$, effectively avoiding the computational bottleneck of the algorithm of Theorem~\ref{thm:main_eps}. 

\begin{lemma}\label{lem:qgram_eps_del_construct}
    Given a database of $n$ documents $\setstr=S_1,\dots, S_{n}$ from $\Sigma^{[1,\ell]}$, and an integer $q\leq \ell$, a data structure satisfying the properties of Lemma~\ref{lem:qgram_eps_delta} that answers $q$-gram counting queries in $O(q)$ time can be computed in $O(n\ell( \log q +\log |\Sigma|))$ time and $O(n\ell)$ space.
\end{lemma}
\begin{proof}
Our main tool is the suffix tree $\ST$ of the string $S=S_1\$_1S_2\$_2\cdots S_n\$_n$, with $\setstr=\{S_1,S_2,\ldots,S_n\}$ and $\$_1,\$_2,\ldots,\$_n\notin\Sigma$, preprocessed to answer weighted ancestor queries in $O(1)$ time. 
A weighted ancestor query consists of two integers $p,\ell>0$, where $p$ is the ID of a leaf of $\ST$; the answer is the farthest (closest to the root) ancestor $v$ of leaf $p$ with string depth at least $\ell$. Any such query can be answered in $O(1)$ time after a linear-time preprocessing to build a linear-space additional data structure~\cite{DBLP:conf/esa/GawrychowskiLN14,DBLP:conf/cpm/BelazzouguiKPR21}. We also store, within each branching node $v$ of $\ST$, the ID $\leaf(v)$ of the leftmost descending leaf. This extra information can be computed in linear time with a DFS of $\ST$, and requires linear extra space. The ID stored at a node $v$ gives, together with the string depth of $v$, a \emph{witness occurrence} in some document from $\setstr$ of the substring represented by $v$.

Our procedure computes an implicit representation of sets $\candidatespruned_{2^k}$ for increasing values of $k\in [0,j]$, where $j=\lfloor\log q\rfloor$; it returns a compacted trie of the elements of $\candidatespruned_{q}$ with a noisy counter for each element. For each phase $k\leq j$, we perform the following steps. (i) Traverse $\ST$ to find the set of branching nodes whose string depth is at least $2^k$ and such that the string depth of their parent node is strictly smaller than $2^k$. We call such nodes \emph{\twokMin}.  (ii) For $k=0$, add noise to the frequency stored in each $1$-minimal node as in Lemma~\ref{lem:qgram_eps_delta}, and mark as belonging to $\candidatespruned_{2^0}$ the nodes whose noisy count exceeds $2\alpha$. In Phase $k>0$, we ask two weighted ancestor queries for each \twokMin node $v$: one with $p=\leaf(v)$ and $\ell=2^{k-1}$, one with $p=\leaf(v)+2^{k-1}$ and $\ell=2^{k-1}$. If the answer to both queries is a node marked as belonging to $\candidatespruned_{2^{k-1}}$, we add noise to the frequency of $v$ and mark it as belonging to $\candidatespruned_{2^{k}}$ if the noisy count exceeds $2\alpha$; we proceed to the next node otherwise. Note that at the end of Phase $k$ we can remove from $\ST$ the markers for $\candidatespruned_{2^{k-2}}$.
Finally, to compute $\candidatespruned_q$, for each $q$-minimal node $v$ we ask weighted ancestor queries with $p=\leaf(v)$ and $\ell=2^{j}$ and with $p=\leaf(v)+q-2^{j}+1$ and $\ell=2^{j}$. If the answer to both queries is a node marked for $\candidatespruned_{2^{j}}$, we add noise to the frequency of $v$; if it exceeds $2\alpha$, we insert 
$\str{v}[1,q]$ in the output trie and store
this noisy counter within the corresponding leaf. 
To answer a query for a $q$-gram $P$, we spell $P$ from the root of the trie in $O(q)$ time. We return its noisy counter if we find it, $0$ otherwise.

\textbf{Time and space analysis.} For each value of $k$, the algorithm traverses $\ST$ in $O(n\ell)$ time; since the \twokMin nodes partition the leaves of $\ST$, it asks up to $2$ $O(1)$-time weighted ancestor queries per leaf, thus requiring $O(n\ell)$ total time. Since $\ST$ can be constructed in $O(n\ell\log|\Sigma|)$ time~\cite{DBLP:conf/focs/Farach97} and $k$ takes $\lfloor\log q\rfloor$ distinct values, the whole procedure requires $O(n\ell( \log q +\log |\Sigma|))$ time.
The space is $O(n\ell)$ because $\ST$ and the additional data structure for weighted ancestor queries occupy $O(n\ell)$ words of space~\cite{DBLP:conf/cpm/BelazzouguiKPR21}, and at any phase $k$, each node stores at most two markers.
Finally, the output compacted trie occupies space $O(n\ell)$ because (i), since $\setstr$ contains at most $O(n\ell)$ distinct $q$-grams it has at most $O(n\ell)$ leaves, and $(ii)$, since all the $q$-grams in the trie occur in $\setstr$, the edge labels can be compactly represented with intervals of positions over $\setstr$, similar to the edge labels of $\ST$.
\end{proof}

Lemmas~\ref{lem:qgram_eps_delta} and~\ref{lem:qgram_eps_del_construct} together give Theorem~\ref{thm:qgrams_eps_delta}.

\thmQgramsEpsdelta*

\section{Counting Functions on Trees}\label{sec:tree_count}

In this section, we first prove Theorem~\ref{thm:tree_count} 
and then show an improvement for $(\epsilon,\delta)$-differential privacy if the sensitivity of the count function for every node is bounded.

 The strategy is essentially the same algorithm as the one applied to $\candidatetrie$ in the previous sections.

 \thmTreeCount*
 
\begin{proof}
    The algorithm starts by computing a heavy path decomposition of $T$. Then, it uses the Laplace mechanism to compute the counts of all heavy path roots with $\epsilon/2$-differential privacy. Then, for every heavy path $p=v_0,\dots, v_{|p|-1}$, we compute the prefix sums of the difference sequence $\diff_p(\setstr)[i]=c(v_i,\setstr)-c(v_{i-1},\setstr)$ via Lemma~\ref{lem:binary_trees} with $\epsilon/2$-differential privacy. Let $r=v_0$ be the root of the heavy path. For any node $v_i$ on heavy path $p$, we can compute $\hat{c}(v_i)$ from $\hat{c}(r)$ plus the approximate value of $\sum_{j=1}^{i} \diff_p(\setstr)$.

    To analyze the accuracy, let $\setstr$ and $\setstr'$ be neighboring and let $l$ be a leaf such that $|c(l,\setstr)-c(l,\setstr')|=b\leq d$. Note that $c(l,\setstr)$ contributes to $c(v,\setstr)$ if and only if $l$ is below $v$. In particular, by Lemma~\ref{lem:heavy_path}, it can contribute at most $b$ to the function $c$ of at most $O(\log |V|)$ heavy path roots. Thus, for the roots $r_1,\dots, r_k$ of heavy paths, we have $\sum_{i=1}^k |c(r_i,\setstr)-c(r_i,\setstr')|=O(d\log |V|)$. Thus, by Lemma~\ref{lem:Laplacemech}, we can estimate $c(r_1,\setstr),\dots, c(r_k,\setstr)$ up to an error of $O(\epsilon^{-1}d \log |V|\ln(k/\beta))$ with probability at least $1-\beta/2$.
    
Similarly, $l$ can only contribute by at most $b$ to the sensitivity of the difference sequence of any heavy path that has a root above $l$. Thus, using Lemma~\ref{lem:binary_trees} with $L=O(d\log |V|)$, we can estimate all prefix sums of $\diff_p$ for all $k$ heavy paths $p$ with an error at most $O(\epsilon^{-1}d\log |V|\log h \log(hk/\beta))$ with probability at least $1-\beta/2$.
%
\end{proof}

We now give an improvement for $(\epsilon,\delta)$-differential privacy if the sensitivity of the count function for every node is bounded.

\thmEpsdelTrees*

\begin{proof}
    The algorithm starts by computing a heavy path decomposition of $T$. Then, it uses the Gaussian mechanism to compute the counts of all heavy path roots with $(\epsilon/2,\delta/2)$-differential privacy. Then, for every heavy path $p=v_0,\dots, v_{|p|-1}$, we compute the prefix sums of the difference sequence $\diff_p(\setstr)[i]=c(v_i,\setstr)-c(v_{i-1},\setstr)$ via Lemma~\ref{lem:binary_trees_epsdel} with $(\epsilon/2,\delta/2)$-differential privacy. For any node $v_i$ on the heavy path, we can compute $\hat{c}(v_i)$ from $\hat{c}(r)$, where $r$ is the root of the heavy path that $v$ lies on, plus the approximate value of $\sum_{j=1}^{i} \diff_p(\setstr)$.

    To analyze the accuracy, let $\setstr$ and $\setstr'$ be neighboring and let $l$ be a leaf such that $|c(l,\setstr)-c(l,\setstr')|=b\leq \min(\Delta,d)$. Note that $c(l,\setstr)$ contributes to $c(v,\setstr)$ if and only if $l$ is below $v$. In particular, it can only contribute to the function $c$ of at most $O(\log |V|)$ heavy path roots by at most $b$, by Lemma~\ref{lem:heavy_path}. Thus, for the roots $r_1,\dots, r_k$ of heavy paths, we have $\sum_{i=1}^k |c(r_i,\setstr)-c(r_i,\setstr')|=O(d\log |V|)$, and $|c(r_i,\setstr)-c(r_i,\setstr')|\leq \Delta$. Thus, by Lemma~\ref{lem:l2norm}, the $L_2$-sensitivity of $c(r_i)$, $i=1,\dots, k$ is bounded by $\sqrt{d\Delta\log |V|}$. By Lemma~\ref{lem:gaussianmech}, we can estimate $c(r_1,\setstr),\dots, c(r_k,\setstr)$ up to an error of $O(\epsilon^{-1}\sqrt{d\Delta \log |V|\ln(2/\delta)\ln(k/\beta)})$ with probability at least $1-\beta/2$.
    
Similarly, $l$ can only contribute by at most $b$ to the sensitivity of the difference sequence of any heavy path that has a root above $l$, and the $L_1$-sensitivity of any one heavy path is additionally bounded by $\Delta$. Thus, using Lemma~\ref{lem:binary_trees_epsdel} with $L=O(d\log |V|)$, we can estimate all prefix sums of $\diff_p$ for all $k$ heavy paths $p$ with an error at most $O(\epsilon^{-1}\sqrt{d\Delta\log |V|\ln(1/\delta)\log(hk/\beta)}\log h )$ with probability at least $1-\beta/2$.
\end{proof}

\section{Lower Bounds}\label{sec:lower_bound}
In this section, we collect our lower bound results. In Section \ref{sec:lower-approx-sub}  we show an $\Omega(\ell)$ lower bound for approximate differentially private \substringcount\ and \substringmining\ via a worst case example. In Section~\ref{sec:lower-approx-doc}, we show a lower bound of roughly $\Omega(\sqrt{\ell})$ for approximate differentially private \documentcount. While these also imply almost tight lower bounds for pure differential privacy, in Section \ref{sec:lower-pure} we show a direct lower bound for $\epsilon$-differential privacy which captures all problems (\documentcount, \substringcount, \qgrammining, \substringmining) with a simple packing argument.

\subsection{Pure Differential Privacy}\label{sec:lower-pure}
We first state a theorem from~\cite{DBLP:books/sp/17/Vadhan17} (Theorem 7.5.14), which immediately implies a lower bound on the additive error of any $\epsilon$-differentially private algorithms solving \documentcount\ and \substringcount. In Theorem~\ref{thm:lower_bound}, we extend this lower bound to also hold for the easier problem where we want to find patterns of a fixed length which have a count above a given threshold (i.e., for \qgrammining). In the following, for any query $q:\chi\rightarrow\{0,1\}$, the corresponding \emph{counting query} $q:\chi^n\rightarrow[0,1]$ is defined as $q(x)=\frac{1}{n}\sum_{i=1}^n q(x_i)$.
  \begin{lemma}[\cite{DBLP:books/sp/17/Vadhan17}, Theorem 7.5.14]\label{lem:lower}
        Let $Q={q:\chi\rightarrow\{0,1\}}$ be any class of counting queries that can distinguish all elements from $\chi$. That is, for all $w\neq w'\in \chi$, there exists a query $q\in Q$ such that $q(w)\neq q(w')$. Suppose $\alg:\chi^n\rightarrow \mathbb{R}^Q$ is an $(\epsilon,\delta)$-differentially private algorithm that with high probability answers every query in $Q$ with error at most $\alpha$. Then 
        \begin{align*}
            \alpha=\Omega\left(\min\left(\frac{\log|\chi|}{n\epsilon},\frac{\log(1/\delta)}{n\epsilon},1/2\right)\right).
        \end{align*}
    \end{lemma}
    For $\epsilon$-differential privacy, $\delta=0$, and therefore the second bound in the minimum is always~$\infty$.
For the universe $X=\Sigma^{\ell}$, let $Q$ be the family of point functions over $X$, i.e., $Q$ contains for every $S\in X$ a query $q_S$ such that $q_S(S')=1$ if and only if $S'=S$, and 0 otherwise. Clearly, $Q$ distinguishes all elements from $\Sigma^{\ell}$. We note that an algorithm solving \documentcount\ and \substringcount\  can answer every query in $Q$ on a data set $X^n$, by reporting the counts of every string of length $\ell$ and dividing the count by $n$. Thus, any $\epsilon$-differentially private algorithm solving  \documentcount\ or \substringcount\  must have an error of $\Omega(\epsilon^{-1}\log X)=\Omega(\min(\epsilon^{-1}\ell \log |\Sigma|,n))$. In the following, we extend this lower bound in two ways:
\begin{enumerate}
    \item We show that it holds even if we only want to output which patterns have a count approximately above a given threshold $\thresh$;
    \item We show that it holds even if we restrict the output to patterns of a fixed length $m$, for any $m\geq 2\log \ell$. 
\end{enumerate}

We prove Theorem~\ref{thm:lower_bound} via a packing argument. For this, we need the following definition and lemma.
\begin{definition}[$k$-neighboring]\label{def:k-neighboring}
    Let $\chi$ be a data universe. Two data sets $\setstr$ and $\setstr'$ are called $k$-neighboring, if there exists a sequence $\setstr=\setstr_1\sim\setstr_2\sim\dots\sim\setstr_{k-1}\sim\setstr_k=\setstr'$.
\end{definition}
\begin{fact}[Group Privacy]\label{fact:k-neighboring}
     Let $\chi$ be a data universe, and $\epsilon>0$. If an algorithm $A:\chi^{*}\rightarrow \range(A)$ is $\epsilon$-differentially private, then for all $k$-neighboring $\setstr\in \chi^*$ and $\setstr'\in \chi^*$ and any set $U\in \range(A)$, we have $
        \Pr[A(\setstr)\in U]\leq e^{k\epsilon}\Pr[A(\setstr')\in U]$.
\end{fact}

\thmLB*

\begin{proof}
    Let $\alg$ be the assumed $\epsilon$-differentially private algorithm with error at most $\alpha$.
    Assume $\alpha<n/2$, else we have $\alpha=\Omega(n)$. Let $B=2\alpha\leq n$.
    Fix two symbols $0,1\in \Sigma$ and let $\widehat{\Sigma}=\Sigma\setminus\{0,1\}$.
     For any set of $k$ patterns $P_1,\dots,P_k$ from $\widehat{\Sigma}^{m/2}$, where $k=\frac{\ell}{m}$, we construct a string $S_{P_1,\dots,P_k}$ as follows: For any $i\leq \ell$, let $c_i$ be the $m/2$-length binary code of $i$. Note that $m\geq 2\lceil\log \ell\rceil$. We define $S_{P_1,\dots,P_k}=P_1c_1P_2c_2\dots P_kc_k$.

    For a string of this form, we define $\setstr(P_1,\dots, P_k)$ as the data set which contains $B$ copies of $S_{P_1,\dots,P_k}$ and $n-B$ copies of $0^{\ell}$. For any two sets of patterns $P_1,\dots,P_k$ from $\widehat{\Sigma}^{m/2}$ and $P'_1,\dots,P'_k$ from $\widehat{\Sigma}^{m/2}$, the data sets $\setstr(P_1,\dots, P_k)$ and  $\setstr(P'_1,\dots, P'_k)$ are $B$-neighboring (see Definition~\ref{def:k-neighboring}).

    Set $\thresh=B/2=\alpha$. 
    We call $\mathcal{E}(P_1,\dots, P_k)$ the event that the output $\mathcal{P}$ of the algorithm contains the strings $P_1c_1$, $P_2c_2$, \dots, $P_kc_k$, and it does not contain any other strings suffixed by $c_1,c_2,\dots,c_k$. For $\setstr(P_1,\dots,P_k)$, we have that $\counts_{\Delta}(P_ic_i,\setstr(P_1,\dots,P_k))=B\geq \tau+\alpha$, for all $i=1,\dots,k$, independent of $\Delta$. Further, any other string with suffix $c_i$ for any $i=1,\dots,k$ has a count of $0\leq \tau-\alpha$. Thus
    \begin{align*}
        \Pr[\alg(\setstr(P_1,\dots, P_k))\in \mathcal{E}(P_1,\dots, P_k)]\geq 2/3.
    \end{align*}
    Since we assume that $\alg$ is $\epsilon$-differentially private, this gives by Fact~\ref{fact:k-neighboring}, for any other set $P'_1,\dots, P'_k$ of patterns from $\widehat{\Sigma}^{m/2}$, 
        \begin{align*}
        \Pr[\alg(\setstr(P'_1,\dots, P'_k))\in \mathcal{E}(P_1,\dots, P_k)]\geq e^{-B\epsilon}2/3.
    \end{align*}
    Note that by construction, $\mathcal{E}(P_1,\dots, P_k)\cap \mathcal{E}(P'_1,\dots, P'_k)=\emptyset$ for $(P_1,\dots, P_k)\neq (P'_1,\dots, P'_k)$. Thus, we have for a fixed set $(P'_1,\dots, P'_k)$ of patterns,
    \begin{align*}
    1&\geq  \sum_{(P_1,\dots,P_k)\in (\widehat{\Sigma}^{m/2})^k}\Pr[\alg(\setstr(P'_1,\dots, P'_k))\in \mathcal{E}(P_1,\dots, P_k)]\\&\geq\sum_{(P_1,\dots,P_k)\in (\widehat{\Sigma}^{m/2})^k}e^{-B\epsilon}2/3=(|\Sigma|-2)^{mk/2}e^{-B\epsilon}2/3.
    \end{align*}
    Solving for $B$, this gives 
    \begin{align*}
        B\geq \frac{1}{\epsilon}\left(\frac{mk}{2}\ln (|\Sigma|-2)+\log (2/3)\right)=\Omega(\epsilon^{-1}\ell\log |\Sigma|).
    \end{align*}
    The theorem follows since $\alpha=B/2$.
    \end{proof}

\subsection{Lower Bound for Approximately Differentially Private Substring Count}\label{sec:lower-approx-sub}

\LBepsdelSubstr*

\begin{proof}
    Let $\setstr=(S_1,\dots,S_n)$ and $\setstr'=(S'_1,\dots,S_n')$ be neighboring data sets where $S_j=a^{\ell}$ and $S'_j= b^{\ell}$. Further, let $S_i=S_i'=b^{\ell}$ for all $i\neq j$. Let $P=a$. We have $\counts(P,\setstr)=\ell$ and $\counts(P,\setstr')=0$. Assume we output an estimate $\counts^*$ with an additive error $\alpha< \ell/2$ with probability at least $1-\beta$. Then we have 
    \begin{align*}
        1-\beta\;\leq \;\Pr(\counts^*(P,\setstr)\geq \ell/2)\;\leq \;e^{\epsilon}\cdot \Pr(\counts^*(P,\setstr')\geq \ell/2)+\delta \;\leq \;e^{\epsilon}\beta + \delta.
    \end{align*}
    This gives
    $$
    \epsilon\geq\ln\left(\frac{1-\beta-\delta}{\beta}\right).
    $$
\end{proof}

Note that $(\epsilon,\delta)$-differential privacy only provides useful privacy guarantees for $\delta=o(1/n)$~\cite{journals/fttcs/DworkR14}. Setting $\delta=o(1/n)$, Theorem~\ref{thm:LB_approx_substring_count} shows that we cannot have an error $O(n)$ for small constant $\epsilon$ and small failure probability. In particular, we get the following corollary by substitution in Equation~\ref{eq:approx_LB}.

\begin{corollary}
 Let $n,\ell\in \mathbb{N}$  and $\Sigma$ an alphabet of size $|\Sigma|\geq 2$.
No $(\epsilon,\delta)$-differentially private algorithm for \substringcount\ which takes as an input any database of $n$ documents from $\Sigma^{\ell}$ and is 
   $\alpha$-accurate with probability at least $1-\beta$ can satisfy any of the following parameter combinations:
    \begin{enumerate}[(i)]
        \item $\alpha=o(\ell)$, $\epsilon\leq 1$, and $\beta$ is an arbitrarily small constant. In particular, in this case, it must hold $\beta\geq \frac{1-\delta}{e+1}$.
        \item $\alpha=o(\ell)$, $\epsilon=O(1)$, and $\beta=O(1/n)$.
    \end{enumerate}
\end{corollary}


The proof of Theorem~\ref{thm:LB_approx_substring_count} can be adapted to prove the same lower bound for $\approxqgram$, even when $q=1$, as stated in Corollary~\ref{cor:LB_approx_qgram_mining}. The corollary is proved by setting a threshold $\tau=\ell$ and proceeding similar to the proof of Theorem~\ref{thm:LB_approx_substring_count}, observing that if $\alpha<\ell/2$, then $P$ should be considered frequent in $\setstr$ but not on $\setstr'$.

\begin{corollary}\label{cor:LB_approx_qgram_mining}
Let $n,\ell\in \mathbb{N}$  and $\Sigma$ an alphabet of size $|\Sigma|\geq 2$.
Any $(\epsilon,\delta)$-differentially private algorithm for $\approxqgram(.,\ell,\tau)$ which takes as an input any database of $n$ documents from $\Sigma^{\ell}$ and is 
$\alpha$-accurate with probability at least $1-\beta$, $\alpha > 0$, $0<\beta<1$, must satisfy either       $\alpha=\Omega(\ell)$
        or 
     $   \epsilon\geq \ln\left(\frac{1-\beta-\delta}{\beta}\right)$.       \end{corollary}

\subsection{Lower Bound for Approximately Differentially Private Document Count}\label{sec:lower-approx-doc}
To prove a lower bound on the error of any $(\epsilon,\delta)$-differentially private algorithm for $\documentcount$, we show a reduction from the 1-way marginals problem, which essentially asks to compute the arithmetic average over each of the columns of a $n\times d$ binary matrix. 

\newcommand{\marginals}[2]{\mathrm{Marginals}(#1,#2)}
\begin{definition}[1-way marginals]\label{def:marginals}
 Let $d,n\in \mathbb{N}$ and $\chi=\{0,1\}^d$. Given a 
database
 $Y=(Y_0,\dots,Y_{n-1})\in \chi^n$, the $\marginals{n}{d}$ problem asks to compute the vector $(q_0(Y),\dots,q_{d-1}(Y))$, where $q_j(Y)=\frac{1}{n}\sum_{i=0}^{n-1} Y_{i}[j]$, for all $j\in[0,d-1]$. 
\end{definition}

\begin{definition}[Accuracy]
 A randomized algorithm $A$ is $(\alpha,\beta)$-accurate for $\marginals{n}{d}$ if, for all 
datasets
 $Y\in\{\{0,1\}^d\}^n$, it outputs $a_0,\dots, a_{d-1}$ such that
$$\Pr[\max_{j\in[0,d-1]}|a_j-q_j(Y)|>\alpha]\leq \beta,$$ where 
the probability is taken over the random coin flips of the algorithm.
\end{definition}

The following lemma follows directly from~\cite[Lemma 3.3]{JainRSS23}, 
which in turn is based on \cite{BunUV18,HardtT10}.
\begin{lemma}[Lower bound for Marginals \cite{BunUV18,HardtT10,JainRSS23}]\label{lem:marginals}
    For all $\epsilon\in(0,1]$, $\delta\in[0,1)$, $\alpha\in(0,1)$, and $d,n\in\mathbb{N}$, every $(\epsilon,\delta)$-DP algorithm that is  $\left(\alpha, \frac{1}{3}\right)$-accurate for $\marginals{n}{d}$ satisfies:
    \begin{itemize}
        \item if $\delta>0$ and $\delta=o\left(\frac 1 n \right)$, then $\alpha=\Omega\left(\frac{\sqrt{d}}{n\epsilon\log d}\right)$.
        \item if $\delta=0$, then $\alpha=\Omega\left(\frac{d}{n\epsilon}\right)$.
    \end{itemize}
\end{lemma}
\LBepsdelDocument*

\begin{proof}
        Let $d,n\in \mathbb{N}$ and $d\geq 2$. We describe a family of reductions parameterized by $b=\min(s,d+1)\in [3,d+1]$. 
   We reduce any instance $Y=(Y_1,\dots, Y_n)$ of $\marginals{n}{d}$ to an instance of \documentcount\ over the alphabet 
   $\Sigma_b=[0, b-2]\cup \{\$\}$ consisting of $n$ documents of length $\ell=d \left(\lceil\log_{b-1} d\rceil+2\right)$.

     \emph{Position gadgets.} For every position $j\in [0,d-1]$, we define a unique code $\mathrm{code}_{b-1}(j)$ of length $\lceil\log_{b-1} d\rceil$ using only letters from $[0, b-2]$. For instance, for $\Sigma_3$, $\mathrm{code}_{2}(j)$ is the binary encoding of $j$; for $\Sigma_{d+1}$,  $\mathrm{code}_{d}(j)=j$. We then define, for every $Y_i$ and every $j$, the position gadget $\mathrm{code}_{b-1}(j)Y_i[j]\$ $.

     \emph{Input to \documentcount.} We construct a string $S_i$ for every $Y_i$. The string $S_i$ is obtained by concatenating all of the position gadgets of $Y_i$ interleaved with the letter $\$$ (note that $\$$ does not appear in any of the position gadgets). We thus have \[S_i=\mathrm{code}_{b-1}(0)Y_i[0]\$\mathrm{code}_{b-1}(1)Y_i[1]\$\cdots \mathrm{code}_{b-1}(d-1)Y_i[d-1]\$\]
     for all $i\in [0,n-1]$.

     \emph{Queries.} To compute $q_j(Y)$, we simply query the pattern $P=\mathrm{code}_{b-1}(j)1$ and divide the resulting count by $n$. This is correct because $q_j(Y)$ is given precisely by the number of indexes $i$ such that $Y_i[j]=1$, divided by $n$. 

Neighboring inputs for $\marginals{n}{d}$ give neighboring inputs to \documentcount\ by construction, as any string $S_i$ depends only on $Y_i$. Since any solution which is $(\alpha,\beta)$-accurate for \documentcount\ is $\left(\frac{\alpha}{n},\beta\right)$-accurate for $\marginals{n}{d}$, by Lemma~\ref{lem:marginals} we obtain the statement.
\end{proof}





\bibliographystyle{plain}
\bibliography{References}
\end{document}